\documentclass[11pt]{article}

\usepackage[margin=1in]{geometry}
\usepackage[T1]{fontenc}
\usepackage{lmodern}
\usepackage{amsthm,amsmath,amsfonts,amssymb}
\usepackage[authoryear,round]{natbib}
\usepackage{graphicx}
\usepackage{bm}
\usepackage{booktabs}
\usepackage{xcolor}
\usepackage{tikz}
\usepackage{url}
\usepackage{placeins}
\usepackage[colorlinks,citecolor=blue,urlcolor=blue,linkcolor=blue]{hyperref}

\theoremstyle{plain}
\newtheorem{proposition}{Proposition}
\newtheorem{propositionS}{Proposition}

\newcommand{\normone}[1]{\lVert #1 \rVert_1}
\newcommand{\normtwo}[1]{\lVert #1 \rVert_2}
\newcommand{\bE}{\bm{E}}
\newcommand{\bx}{\bm{x}}
\newcommand{\bw}{\bm{w}}

\newcommand{\btheta}{\bm{\theta}}
\newcommand{\bphi}{\bm{\phi}}
\newcommand{\bpsi}{\bm{\psi}}
\title{Self-thresholding angular-radial models of extremal dependence}
\author{Manuele Leonelli\thanks{School of Science and Technology, IE University, Madrid, Spain. Email: \href{mailto:manuele.leonelli@ie.edu}{manuele.leonelli@ie.edu}}
\and Dani Gamerman\thanks{Department of Statistical Methods, Institute of Mathematics, Federal University of Rio de Janeiro, Rio de Janeiro, Brazil. Email: \href{mailto:dani@im.ufrj.br}{dani@im.ufrj.br}}}
\date{}

\begin{document}

\maketitle

\begin{abstract}
Assessing compound risks requires a model for routine and extreme joint behaviour together, and a principled answer to where the joint tail begins. We propose the self-thresholding angular-radial (STAR) model: after marginal standardization, the radius given the angle follows a spliced distribution, with a flexible body below an angle-dependent threshold curve and a generalized Pareto tail above it. The threshold curve is estimated with all other parameters, so its uncertainty propagates to every summary, and the joint density is in closed form because the splice normalizes in one dimension for every angle. The tail scale function is the reciprocal gauge of the implied limit set, and the coefficient of extremal association is an explicit parameter, so a prior with atoms at tail independence and at the asymptotic dependence boundary yields posterior probabilities of the dependence class. At that boundary the copula of the model is asymptotically dependent, while its margins and joint tail are both lighter than exponential by a factor inversely proportional to the level. Simulations show accurate recovery of finite-level dependence summaries within and beyond the model class, and the slow separation of asymptotic dependence from strong asymptotic independence that the theory predicts. For warm-season daily maxima of fine particulate matter and ozone at a Beijing station, the pair is asymptotically independent with posterior probability $0.96$ and both pollutants exceed their $0.9$ quantiles on about two and a half days per season; the model gives summaries comparable to those of existing approaches without requiring the arbitrary choice of a threshold. The construction extends to $d$ dimensions through pairwise deficits of the gauge, illustrated by adding coarse particulate matter as a third pollutant.
\end{abstract}


\section{Introduction}\label{sec:intro}

Many questions about multivariate risk concern the body of a joint distribution and its tail at once: a flood defence can fail under combinations of river flow and sea level in which neither is record-breaking, and the impact of a heatwave depends on temperature and humidity together, when only one may be exceptional. Such questions require a model that describes the body and the joint tail within one framework, and a defensible answer to where the tail begins. In the univariate setting, extreme value mixture models combine a flexible body density with a generalized Pareto distribution (GPD) for threshold exceedances \citep{behrens2004,nascimento2012}, with a parallel line embedding GPD tail behaviour in a single parametric family \citep{naveau2016}; see \citet{scarrott2012} and \citet{naveau2026} for reviews. Their unifying idea is that the location of the tail is itself an object of inference.

Carrying this idea to the multivariate setting raises a difficulty absent from the univariate case: the threshold is no longer a point but a boundary in $\mathbb{R}^d$, and its shape encodes a decision about which observations count as jointly extreme, whether those in a rectangular region, those beyond a radial level, or those exceeding a more general risk functional. The visual diagnostics of univariate practice do not transfer \citep{belzile2026}, applied work defaults to a common marginal quantile for every variable \citep{davisonhuser2015}, and the univariate threshold selection procedures of \citet{northrop2017} and \citet{murphy2025} have no multivariate counterpart.

Models for the body and tail of multivariate data follow three routes. The first blends dependence structures over the full support, through generalized Pareto copulas pieced onto a body copula \citep{aulbach2012a}, spliced margins with copula mixtures that deliver posterior statements about the dependence class \citep{leonelli2020}, and dynamically weighted mixtures of a body and a tail copula \citep{andre2024}. These avoid a hard threshold, but because both components act on the full support the body influences the fitted tail and the tail remains sensitive to central observations \citep{hu2026}. The second separates body and tail on disjoint regions, with the tail governed by the multivariate generalized Pareto distribution (mGPD) of \citet{rootzen2006}: \citet{hu2026} combine a multivariate normal body with an mGPD above an estimated threshold vector, with a spatial analogue in \citet{cuba2024}. The separation removes the entanglement, at the price of a rectangular exceedance region, asymptotic dependence by construction, and a density discontinuous across the boundary. The third avoids a threshold altogether, extending the single-family construction of \citet{naveau2016} to several dimensions \citep{alotaibi2025,gaetan2025}, but offers neither likelihood-based Bayesian inference nor a boundary marking where asymptotic theory applies.

A separate and very active literature models the joint tail alone. The convergence of scaled sample clouds onto a limit set links its shape to the classical descriptions of extremal dependence \citep{nolde2014,nolde2022}, and inference has followed for the limit set boundary \citep{wadsworthcampbell2024,simpson2024,majumder2025}, for radial generalized Pareto distributions and return sets \citep{papastathopoulos2023} and for semiparametric angular-radial densities \citep{mackay2023,murphybarltrop2024}, alongside the conditional approach of \citet{heffernan2004}. None models the distribution below the threshold, which is fixed in advance, typically at an empirical quantile, so threshold uncertainty is excluded.

This paper brings the two literatures together with the self-thresholding angular-radial (STAR) model, a generative angular-radial model that is self-thresholding in the tradition of \citet{behrens2004}: the radius conditional on the angle follows a spliced distribution, with a flexible body density below a threshold curve and a GPD above it. The threshold curve, an interpretable power-mean contour with the rectangle of \citet{hu2026} as a limit, is estimated jointly with all other unknowns, so its uncertainty propagates into every derived quantity. Because the splice normalizes in one dimension for every angle, the joint density is in closed form. The fitted tail scale function is the reciprocal gauge of the limit set of \citet{nolde2022}, so the residual tail dependence coefficient of \citet{ledford1996} and the angular dependence function of \citet{wadsworthtawn2013} follow from their limit-set characterizations \citep{nolde2014,nolde2022}. Under the working scale family, the coefficient of extremal association of \citet{coles1999} is an explicit parameter, so that a three-component prior with atoms at tail independence and at the asymptotic dependence boundary returns posterior probabilities of the dependence class and a posterior for the strength of association.
Beyond the limit set, we derive the exact joint and marginal tail decay of the model: both carry a factor of order $1/t$, so the copula is asymptotically dependent at the boundary, while on the standardized scale of the data $\chi(q)$ vanishes there logarithmically (Section \ref{sec:jointdecay}); extensions with $\chi(q)$ bounded away from zero, through an explicit asymptotically dependent component, are discussed in Section \ref{sec:discussion}. In $d$ dimensions the splice and the threshold carry over, and the scale function subtracts pairwise deficits from the gauge of independence (Section \ref{sec:multi}); a fit to three pollutants shows what this adds to pairwise analyses and what remains open (Section \ref{sec:app3}). Simulations and an application to fine particulate matter and ozone in Beijing compare the model with the conditional extremes model \citep{heffernan2004} and the geometric model \citep{wadsworthcampbell2024}.

\section{A self-thresholding angular-radial model}\label{sec:model}

Throughout, $\bE=(E_1,\ldots,E_d)$ denotes a random vector with standard exponential margins, obtained from the vector of interest $\bm{Y}$ through $E_j=-\log\{1-F_j(Y_j)\}$ \citep{coles2001}; estimation of the margins is treated in Section \ref{sec:inference}.

\subsection{Angular-radial decomposition}\label{sec:polar}
Dependence is modelled in pseudo-polar coordinates $R=\normone{\bE}=\sum_{j=1}^d E_j$ and $\bm{W}=\bE/R$, so that $R>0$ and $\bm{W}$ lies in the unit simplex $\mathcal{S}_{d-1}=\{\bw\in[0,1]^d:\sum_j w_j=1\}$. The map $\bE\mapsto(R,\bm{W})$ has Jacobian $r^{d-1}$, so that any joint density factorizes as
\begin{equation}\label{eq:factorization}
f_{\bE}(\bm{e})=f_{\bm{W}}(\bw)\,f_{R\mid\bm{W}}(r\mid\bw)\,r^{-(d-1)}, \qquad r=\normone{\bm{e}},\ \bw=\bm{e}/r.
\end{equation}
The model specifies the two factors on the right of \eqref{eq:factorization} directly: a flexible angular density $f_{\bm{W}}$ describing the directions in which observations occur, and a conditional radial density $f_{R\mid\bm{W}}$ describing how far observations travel in each direction. For clarity of exposition we present the bivariate case, identifying $\bm{W}$ with the scalar $W=E_1/R\in(0,1)$, and indicate the ingredients of the general case where they differ; Section \ref{sec:multi} gives the scale function and the dependence results in $d$ dimensions. This decomposition also underlies the SPAR model of \citet{mackay2023} and the geometric inference of \citet{wadsworthcampbell2024}.

\subsection{The spliced radial model}\label{sec:splice}

Conditional on the angle, the radius follows a univariate spliced distribution, with every ingredient allowed to vary smoothly with the angle. Let $\rho(w)>0$ denote the threshold function, $g_{\mathrm{b}}(\cdot\mid w)$ a body density on $(0,\infty)$ with distribution function $G_{\mathrm{b}}(\cdot\mid w)$, and $h_{\mathrm{GP}}(\cdot;\sigma,\xi)$ the density of the generalized Pareto distribution with scale $\sigma>0$ and shape $\xi$. The conditional radial density is
\begin{equation}\label{eq:splice}
f_{R\mid W}(r\mid w)=
\begin{cases}
g_{\mathrm{b}}(r\mid w), & 0<r\le\rho(w),\\[4pt]
\{1-G_{\mathrm{b}}(\rho(w)\mid w)\}\,h_{\mathrm{GP}}\{r-\rho(w);\sigma(w),\xi(w)\}, & r>\rho(w).
\end{cases}
\end{equation}
The construction integrates to one for every angle, since the body contributes mass $\varpi(w)=G_{\mathrm{b}}(\rho(w)\mid w)$ and the tail the remaining $1-\varpi(w)$. This is the central computational property of the model: a patchwork of body and tail densities separated by a curved boundary on the observation scale involves an intractable normalizing constant, the integral of the body density over the sub-threshold region, whereas here it collapses to the one-dimensional integral $G_{\mathrm{b}}(\rho(w)\mid w)$, available in closed form for standard body densities. The joint density follows from \eqref{eq:factorization} with no further normalization. The density \eqref{eq:splice} is in general discontinuous at $r=\rho(w)$; continuity could be enforced by tying $\sigma(w)$ to the body density at the threshold, as in \citet{carreau2009}, but this would conflict with the structural scale \eqref{eq:scalefam}. The discontinuity does not carry over to prediction: with a random threshold curve the posterior predictive density averages over $\rho(\cdot)$ and is smooth \citep{nascimento2012}.

For the body we take a Gamma density with angle-dependent rate,
\begin{equation}\label{eq:body}
g_{\mathrm{b}}(r\mid w)=\frac{\lambda(w)^{\alpha}}{\Gamma(\alpha)}\,r^{\alpha-1}\exp\{-\lambda(w)r\},
\end{equation}
a choice motivated by two considerations. First, under independence in the bivariate case, $R\sim\mathrm{Gamma}(2,1)$ independently of $W\sim\mathrm{Uniform}(0,1)$, so with $\alpha=2$, a constant rate $\lambda\equiv1$ and a uniform angular density the model reproduces independence exactly in the body and to first order in the tail, since the excess distribution of a Gamma converges to the unit-rate exponential as the threshold grows; in dimension $d$ the corresponding values are $\alpha=d$ and $\lambda\equiv1$. Second, by the geometric theory of Section \ref{sec:theory}, light-tailed sample clouds decay radially at direction-dependent exponential rates, which determine their limit set, introduced in Section \ref{sec:parfun} \citep{nolde2022}. Other body families, such as generalized Gamma densities or mixtures of Gamma densities \citep{nascimento2012}, fit in \eqref{eq:splice} without modification. Simulations with a body shape varying across the angle show that a constant shape suffices (Section~\ref{supp:sens}).

\subsection{Threshold curves}\label{sec:threshold}

The threshold function $\rho(\cdot)$ traces a curve in the positive orthant separating the body from the tail. Any positive function on the simplex defines a valid star-shaped boundary. We adopt a parametric family, and a nonparametric alternative would expand $\log\rho(w)$ in the Bernstein basis of \eqref{eq:parfun}. For anchor levels $u_1,u_2>0$ and exponent $p>0$, substituting $x_1=rw$ and $x_2=r(1-w)$ into the power-mean contour $(x_1/u_1)^p+(x_2/u_2)^p=1$ and solving for $r$ gives
\begin{equation}\label{eq:rhop}
\rho(w)=\Big\{\Big(\frac{w}{u_1}\Big)^{p}+\Big(\frac{1-w}{u_2}\Big)^{p}\Big\}^{-1/p}.
\end{equation}
The exponent, estimated with all other parameters, controls the geometry of the extremal region (Figure \ref{fig:thresholds}): for $p=1$ the boundary is the straight line joining $(u_1,0)$ and $(0,u_2)$, for $p=2$ an elliptical arc, and values $p\in(0,1)$ give boundaries concave towards the axes. As $p\to\infty$, $\rho(w)\to\min\{u_1/w,\,u_2/(1-w)\}$, whose exceedance region $\{\bx:\ x_1>u_1\ \text{or}\ x_2>u_2\}$ is the rectangular region of \citet{hu2026}; above it they fit an mGPD, whereas \eqref{eq:splice} fits a generalized Pareto tail to the radius in each direction. We recommend $p\ge1$, as a concave curve can sag into the body along the diagonal and defeat the anchoring priors of Section \ref{sec:priors}. In $d$ dimensions the family is $\rho(\bw)=\{\sum_j(w_j/u_j)^p\}^{-1/p}$, with $d$ anchors and one exponent, and generalizes without change (Supplementary Section~\ref{supp:multi-comp}).

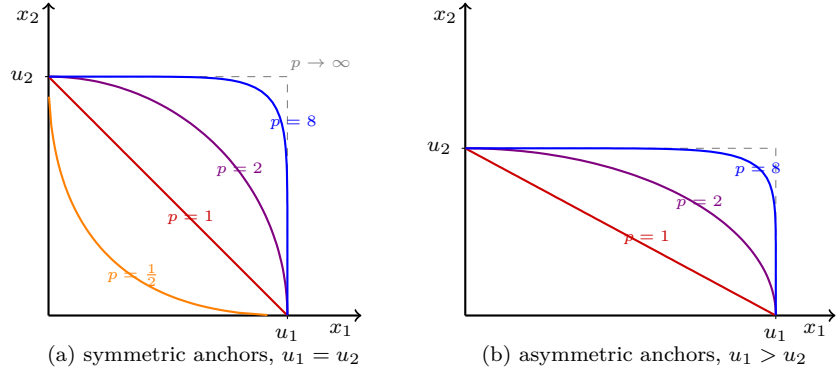
\begin{figure}[t]
\centering
\begin{tikzpicture}[scale=1.58]
\draw[->,thick] (0,0) -- (2.62,0) node[below left=-1pt] {\scriptsize $x_1$};
\draw[->,thick] (0,0) -- (0,2.62) node[below left=-1pt] {\scriptsize $x_2$};
\draw (2,0.03) -- (2,-0.03) node[below] {\scriptsize $u_1$};
\draw (0.03,2) -- (-0.03,2) node[left] {\scriptsize $u_2$};
\draw[dashed,gray] (2,0) -- (2,2) -- (0,2);
\draw[thick,orange] plot coordinates {(0.004,1.832) (0.023,1.591) (0.041,1.470) (0.057,1.383) (0.072,1.314) (0.086,1.256) (0.100,1.205) (0.114,1.160) (0.127,1.120) (0.139,1.083) (0.152,1.049) (0.164,1.018) (0.176,0.989) (0.188,0.961) (0.200,0.935) (0.212,0.910) (0.223,0.887) (0.235,0.864) (0.246,0.843) (0.258,0.822) (0.269,0.802) (0.280,0.783) (0.291,0.765) (0.303,0.747) (0.314,0.729) (0.325,0.713) (0.336,0.696) (0.348,0.680) (0.359,0.665) (0.370,0.649) (0.382,0.634) (0.393,0.620) (0.404,0.606) (0.416,0.592) (0.428,0.578) (0.439,0.564) (0.451,0.551) (0.463,0.538) (0.475,0.525) (0.488,0.513) (0.500,0.500) (0.513,0.488) (0.525,0.475) (0.538,0.463) (0.551,0.451) (0.564,0.439) (0.578,0.428) (0.592,0.416) (0.606,0.404) (0.620,0.393) (0.634,0.382) (0.649,0.370) (0.665,0.359) (0.680,0.348) (0.696,0.336) (0.713,0.325) (0.729,0.314) (0.747,0.303) (0.765,0.291) (0.783,0.280) (0.802,0.269) (0.822,0.258) (0.843,0.246) (0.864,0.235) (0.887,0.223) (0.910,0.212) (0.935,0.200) (0.961,0.188) (0.989,0.176) (1.018,0.164) (1.049,0.152) (1.083,0.139) (1.120,0.127) (1.160,0.114) (1.205,0.100) (1.256,0.086) (1.314,0.072) (1.383,0.057) (1.470,0.041) (1.591,0.023) (1.832,0.004)};
\draw[thick,red!80!black] plot coordinates {(0.004,1.996) (0.029,1.971) (0.054,1.946) (0.079,1.921) (0.104,1.896) (0.129,1.871) (0.153,1.847) (0.178,1.822) (0.203,1.797) (0.228,1.772) (0.253,1.747) (0.278,1.722) (0.303,1.697) (0.328,1.672) (0.353,1.647) (0.378,1.623) (0.402,1.598) (0.427,1.573) (0.452,1.548) (0.477,1.523) (0.502,1.498) (0.527,1.473) (0.552,1.448) (0.577,1.423) (0.602,1.398) (0.626,1.373) (0.651,1.349) (0.676,1.324) (0.701,1.299) (0.726,1.274) (0.751,1.249) (0.776,1.224) (0.801,1.199) (0.826,1.174) (0.851,1.149) (0.875,1.125) (0.900,1.100) (0.925,1.075) (0.950,1.050) (0.975,1.025) (1.000,1.000) (1.025,0.975) (1.050,0.950) (1.075,0.925) (1.100,0.900) (1.124,0.876) (1.149,0.851) (1.174,0.826) (1.199,0.801) (1.224,0.776) (1.249,0.751) (1.274,0.726) (1.299,0.701) (1.324,0.676) (1.349,0.651) (1.373,0.627) (1.398,0.602) (1.423,0.577) (1.448,0.552) (1.473,0.527) (1.498,0.502) (1.523,0.477) (1.548,0.452) (1.573,0.427) (1.598,0.402) (1.622,0.378) (1.647,0.353) (1.672,0.328) (1.697,0.303) (1.722,0.278) (1.747,0.253) (1.772,0.228) (1.797,0.203) (1.822,0.178) (1.847,0.153) (1.871,0.129) (1.896,0.104) (1.921,0.079) (1.946,0.054) (1.971,0.029) (1.996,0.004)};
\draw[thick,violet] plot coordinates {(0.004,2.000) (0.029,2.000) (0.055,1.999) (0.082,1.998) (0.109,1.997) (0.137,1.995) (0.166,1.993) (0.195,1.990) (0.225,1.987) (0.255,1.984) (0.287,1.979) (0.319,1.974) (0.351,1.969) (0.385,1.963) (0.419,1.956) (0.453,1.948) (0.488,1.939) (0.524,1.930) (0.561,1.920) (0.598,1.909) (0.635,1.896) (0.674,1.883) (0.712,1.869) (0.751,1.854) (0.790,1.837) (0.830,1.820) (0.870,1.801) (0.910,1.781) (0.950,1.760) (0.990,1.738) (1.031,1.714) (1.071,1.689) (1.111,1.663) (1.150,1.636) (1.190,1.608) (1.229,1.578) (1.267,1.547) (1.305,1.516) (1.342,1.483) (1.379,1.449) (1.414,1.414) (1.449,1.379) (1.483,1.342) (1.516,1.305) (1.547,1.267) (1.578,1.229) (1.608,1.190) (1.636,1.150) (1.663,1.111) (1.689,1.071) (1.714,1.031) (1.738,0.990) (1.760,0.950) (1.781,0.910) (1.801,0.870) (1.820,0.830) (1.837,0.790) (1.854,0.751) (1.869,0.712) (1.883,0.674) (1.896,0.635) (1.909,0.598) (1.920,0.561) (1.930,0.524) (1.939,0.488) (1.948,0.453) (1.956,0.419) (1.963,0.385) (1.969,0.351) (1.974,0.319) (1.979,0.287) (1.984,0.255) (1.987,0.225) (1.990,0.195) (1.993,0.166) (1.995,0.137) (1.997,0.109) (1.998,0.082) (1.999,0.055) (2.000,0.029) (2.000,0.004)};
\draw[thick,blue] plot coordinates {(0.004,2.000) (0.029,2.000) (0.055,2.000) (0.082,2.000) (0.109,2.000) (0.137,2.000) (0.166,2.000) (0.196,2.000) (0.226,2.000) (0.257,2.000) (0.290,2.000) (0.323,2.000) (0.357,2.000) (0.392,2.000) (0.428,2.000) (0.465,2.000) (0.504,2.000) (0.543,2.000) (0.584,2.000) (0.627,2.000) (0.670,2.000) (0.715,2.000) (0.762,2.000) (0.810,2.000) (0.860,2.000) (0.912,2.000) (0.966,1.999) (1.021,1.999) (1.079,1.998) (1.138,1.997) (1.200,1.996) (1.264,1.994) (1.329,1.990) (1.396,1.986) (1.464,1.979) (1.533,1.969) (1.600,1.955) (1.666,1.935) (1.728,1.909) (1.784,1.876) (1.834,1.834) (1.876,1.784) (1.909,1.728) (1.935,1.666) (1.955,1.600) (1.969,1.533) (1.979,1.464) (1.986,1.396) (1.990,1.329) (1.994,1.264) (1.996,1.200) (1.997,1.138) (1.998,1.079) (1.999,1.021) (1.999,0.966) (2.000,0.912) (2.000,0.860) (2.000,0.810) (2.000,0.762) (2.000,0.715) (2.000,0.670) (2.000,0.627) (2.000,0.584) (2.000,0.543) (2.000,0.504) (2.000,0.465) (2.000,0.428) (2.000,0.392) (2.000,0.357) (2.000,0.323) (2.000,0.290) (2.000,0.257) (2.000,0.226) (2.000,0.196) (2.000,0.166) (2.000,0.137) (2.000,0.109) (2.000,0.082) (2.000,0.055) (2.000,0.029) (2.000,0.004)};
\node[font=\tiny,orange] at (0.71,0.34) {$p=\tfrac12$};
\node[font=\tiny,red!80!black] at (1.19,0.82) {$p=1$};
\node[font=\tiny,violet] at (1.60,1.22) {$p=2$};
\node[font=\tiny,blue] at (2.05,1.62) {$p=8$};
\node[font=\tiny,gray] at (2.28,2.10) {$p\to\infty$};
\node at (1.3,-0.34) {\scriptsize (a) symmetric anchors, $u_1=u_2$};
\end{tikzpicture}\hspace{0.5cm}
\begin{tikzpicture}[scale=1.58]
\draw[->,thick] (0,0) -- (3.1,0) node[below left=-1pt] {\scriptsize $x_1$};
\draw[->,thick] (0,0) -- (0,2.62) node[below left=-1pt] {\scriptsize $x_2$};
\draw (2.6,0.03) -- (2.6,-0.03) node[below] {\scriptsize $u_1$};
\draw (0.03,1.4) -- (-0.03,1.4) node[left] {\scriptsize $u_2$};
\draw[dashed,gray] (2.6,0) -- (2.6,1.4) -- (0,1.4);
\draw[thick,red!80!black] plot coordinates {(0.003,1.398) (0.020,1.389) (0.038,1.379) (0.056,1.370) (0.074,1.360) (0.093,1.350) (0.111,1.340) (0.130,1.330) (0.149,1.320) (0.169,1.309) (0.188,1.299) (0.208,1.288) (0.228,1.277) (0.248,1.266) (0.269,1.255) (0.289,1.244) (0.311,1.233) (0.332,1.221) (0.353,1.210) (0.375,1.198) (0.397,1.186) (0.420,1.174) (0.443,1.162) (0.466,1.149) (0.489,1.137) (0.513,1.124) (0.537,1.111) (0.561,1.098) (0.586,1.085) (0.611,1.071) (0.636,1.058) (0.662,1.044) (0.688,1.030) (0.714,1.016) (0.741,1.001) (0.768,0.986) (0.796,0.972) (0.824,0.957) (0.852,0.941) (0.881,0.926) (0.910,0.910) (0.940,0.894) (0.970,0.878) (1.000,0.861) (1.031,0.845) (1.063,0.828) (1.095,0.810) (1.128,0.793) (1.161,0.775) (1.194,0.757) (1.228,0.739) (1.263,0.720) (1.298,0.701) (1.334,0.682) (1.371,0.662) (1.408,0.642) (1.445,0.622) (1.484,0.601) (1.523,0.580) (1.562,0.559) (1.603,0.537) (1.644,0.515) (1.685,0.492) (1.728,0.470) (1.771,0.446) (1.816,0.422) (1.860,0.398) (1.906,0.374) (1.953,0.348) (2.000,0.323) (2.049,0.297) (2.098,0.270) (2.149,0.243) (2.200,0.215) (2.252,0.187) (2.306,0.158) (2.361,0.129) (2.416,0.099) (2.473,0.068) (2.531,0.037) (2.590,0.005)};
\draw[thick,violet] plot coordinates {(0.003,1.400) (0.021,1.400) (0.039,1.400) (0.057,1.400) (0.076,1.399) (0.096,1.399) (0.116,1.399) (0.137,1.398) (0.158,1.397) (0.180,1.397) (0.202,1.396) (0.225,1.395) (0.249,1.394) (0.273,1.392) (0.298,1.391) (0.323,1.389) (0.349,1.387) (0.376,1.385) (0.404,1.383) (0.432,1.380) (0.462,1.378) (0.492,1.375) (0.523,1.371) (0.554,1.368) (0.587,1.364) (0.620,1.360) (0.654,1.355) (0.690,1.350) (0.726,1.344) (0.763,1.338) (0.801,1.332) (0.840,1.325) (0.880,1.317) (0.921,1.309) (0.962,1.301) (1.005,1.291) (1.049,1.281) (1.094,1.270) (1.139,1.258) (1.185,1.246) (1.233,1.233) (1.281,1.218) (1.329,1.203) (1.379,1.187) (1.429,1.170) (1.479,1.151) (1.530,1.132) (1.581,1.112) (1.632,1.090) (1.683,1.067) (1.735,1.043) (1.785,1.018) (1.836,0.991) (1.886,0.964) (1.935,0.935) (1.984,0.905) (2.031,0.874) (2.078,0.842) (2.122,0.809) (2.166,0.775) (2.207,0.740) (2.247,0.704) (2.285,0.668) (2.321,0.631) (2.355,0.593) (2.387,0.555) (2.416,0.517) (2.443,0.479) (2.468,0.440) (2.491,0.402) (2.511,0.364) (2.529,0.326) (2.544,0.288) (2.558,0.250) (2.570,0.213) (2.579,0.177) (2.587,0.141) (2.593,0.106) (2.597,0.072) (2.599,0.038) (2.600,0.005)};
\draw[thick,blue] plot coordinates {(0.003,1.400) (0.021,1.400) (0.039,1.400) (0.057,1.400) (0.076,1.400) (0.096,1.400) (0.116,1.400) (0.137,1.400) (0.158,1.400) (0.180,1.400) (0.203,1.400) (0.226,1.400) (0.250,1.400) (0.274,1.400) (0.300,1.400) (0.326,1.400) (0.353,1.400) (0.380,1.400) (0.409,1.400) (0.439,1.400) (0.469,1.400) (0.501,1.400) (0.533,1.400) (0.567,1.400) (0.602,1.400) (0.639,1.400) (0.676,1.400) (0.715,1.400) (0.756,1.400) (0.798,1.400) (0.842,1.400) (0.887,1.400) (0.935,1.400) (0.984,1.400) (1.036,1.400) (1.090,1.400) (1.146,1.400) (1.205,1.400) (1.267,1.399) (1.331,1.399) (1.399,1.399) (1.470,1.398) (1.544,1.397) (1.621,1.396) (1.702,1.394) (1.787,1.391) (1.874,1.387) (1.963,1.381) (2.054,1.372) (2.143,1.359) (2.230,1.341) (2.310,1.317) (2.381,1.285) (2.441,1.247) (2.489,1.202) (2.525,1.152) (2.551,1.097) (2.568,1.041) (2.580,0.983) (2.588,0.926) (2.593,0.869) (2.596,0.813) (2.598,0.759) (2.599,0.706) (2.599,0.655) (2.600,0.605) (2.600,0.556) (2.600,0.509) (2.600,0.464) (2.600,0.420) (2.600,0.377) (2.600,0.335) (2.600,0.294) (2.600,0.254) (2.600,0.216) (2.600,0.179) (2.600,0.142) (2.600,0.107) (2.600,0.072) (2.600,0.038) (2.600,0.005)};
\node[font=\tiny,red!80!black] at (1.52,0.65) {$p=1$};
\node[font=\tiny,violet] at (1.96,0.95) {$p=2$};
\node[font=\tiny,blue] at (2.45,1.22) {$p=8$};
\node at (1.5,-0.34) {\scriptsize (b) asymmetric anchors, $u_1>u_2$};
\end{tikzpicture}%
\caption{Power-mean threshold curves \eqref{eq:rhop}, beyond which an observation is extreme, for anchors (a) $u_1=u_2$ and (b) $u_1>u_2$. The anchors fix where the curve meets the axes and the exponent its shape, a straight line at $p=1$ approaching the boundary of the componentwise region $\{x_1>u_1\ \text{or}\ x_2>u_2\}$ as $p$ grows.}
\label{fig:thresholds}
\end{figure}

\subsection{Angle-dependent body and tail parameters}\label{sec:parfun}

The body rate is represented on the logarithmic scale in a low-order Bernstein basis,
\begin{equation}\label{eq:parfun}
\log\lambda(w)=\sum_{k=0}^{K_\lambda} b_k\binom{K_\lambda}{k}w^{k}(1-w)^{K_\lambda-k},
\end{equation}
which guarantees positivity and smoothness with interpretable coefficients.

The tail scale carries the extremal dependence of the model and is parametrized so that its strength appears as an explicit parameter. With $\tau(w)=\min(w,1-w)/\max(w,1-w)$, which increases from zero at the axes to one at the diagonal, we set
\begin{equation}\label{eq:scalefam}
\sigma(w)=1+\bar{\chi}\,\tau(w)^{1+\nu(w)}, \quad \bar{\chi}\in[0,1], \quad \nu(w)=\exp\Big\{\sum_{k=0}^{K_\nu} z_k\binom{K_\nu}{k}w^{k}(1-w)^{K_\nu-k}\Big\},
\end{equation}
so that the scalar $\bar{\chi}$ sets the scale at the diagonal, $\sigma(1/2)=1+\bar{\chi}$, and the positive function $\nu(\cdot)$ controls how fast the scale relaxes towards its marginal value of one at the axes. On exponential margins the sample cloud divided by $\log n$ converges, as $n$ grows, onto a fixed shape, the limit set of \citet{nolde2014}, whose boundary lies at distance $\sigma(w)$ from the origin in direction $(w,1-w)$ (Section \ref{sec:limitset}). The shape of the limit set summarizes the extremal dependence: one that bulges along the diagonal describes components that are often large together, and one that hugs the triangle $x_1+x_2\le1$ describes rare joint extremes. Figure \ref{fig:scalefam} shows the roles of $\bar{\chi}$ and $\nu(\cdot)$ in the scale profile and in the limit set. Four properties hold for every admissible parameter value, so no constrained sampling is required. First, $\tau$ vanishes at the axes, so $\sigma(0)=\sigma(1)=1$, which with the envelope below ensures marginal compatibility: radial decay along an axis is marginal exponential decay. Second, the envelope $\sigma(w)\le\min\{1/w,1/(1-w)\}$, which keeps the limit set inside the unit square (Section \ref{sec:limitset}), holds because $\bar{\chi}\,\tau(w)^{\nu(w)}\le1$, with equality on $(0,1)$ only at the diagonal when $\bar{\chi}=1$. Third, $\sigma\ge1$, so the implied limit set contains that of independence: the family models nonnegative extremal association, and negative association is deferred to Section \ref{sec:discussion}. Fourth, $\bar{\chi}$ is exactly the coefficient of extremal association of \citet{coles1999} (Proposition \ref{prop:eta}), hence the notation. Section \ref{sec:multi} extends it to $d$ dimensions.

\begin{figure}[t]
\centering
\input{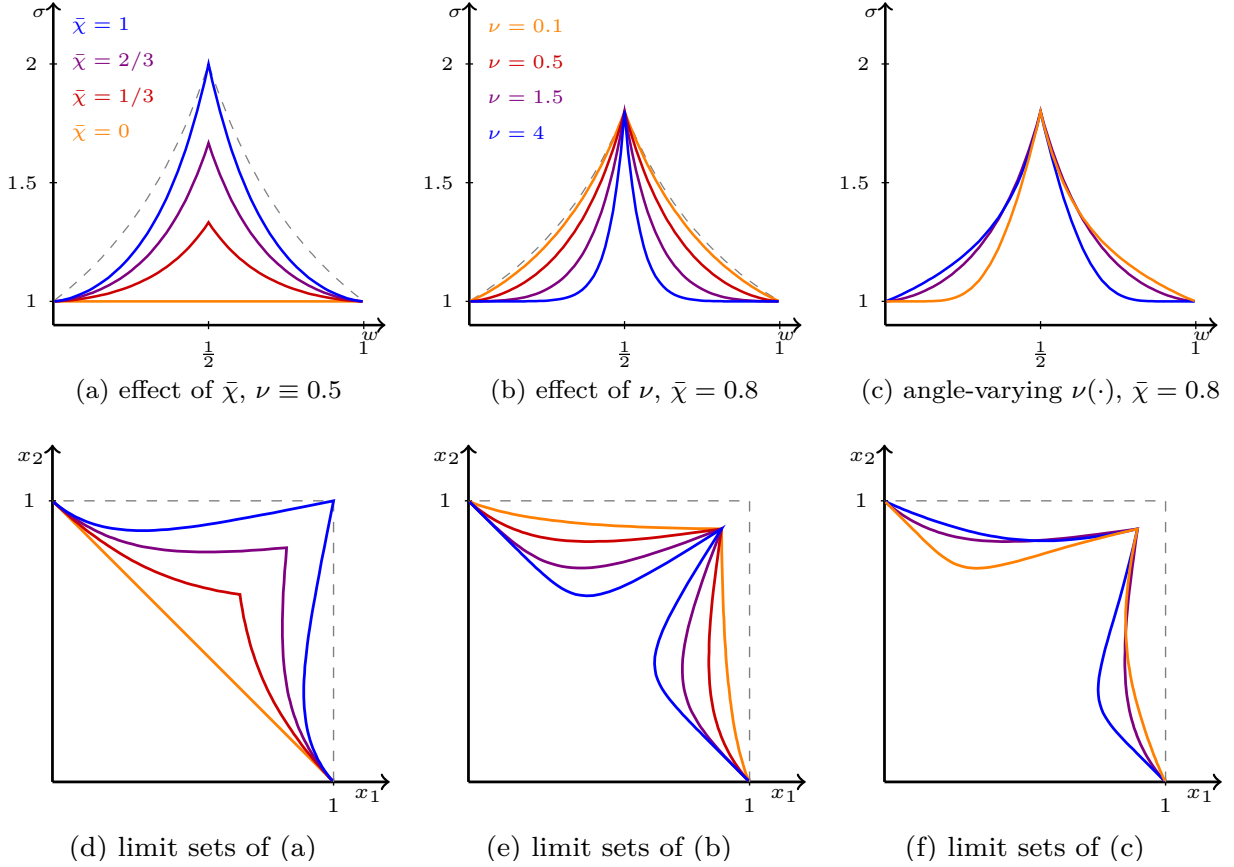}%
\caption{The scale family \eqref{eq:scalefam}: profiles $\sigma(w)$ (top) and limit set boundaries (bottom). Left: $\nu\equiv0.5$ and $\bar{\chi}\in\{0,1/3,2/3,1\}$; raising $\bar{\chi}$ lifts the scale at the diagonal and pushes the limit set towards $(1,1)$, reached when $\bar{\chi}=1$, with the envelope $1+\tau(w)$ dashed (the unit box in the limit set panels). Centre: $\bar{\chi}=0.8$ and $\nu\in\{0.1,0.5,1.5,4\}$; small values of $\nu$ push the profile towards the cap $1+\bar{\chi}\tau(w)$ (dashed), spreading extremal association across directions, and large values concentrate it near the diagonal. Right: $\bar{\chi}=0.8$ with constant (violet), increasing (blue) and decreasing (orange) $\nu(\cdot)$, giving asymmetric profiles.}
\label{fig:scalefam}
\end{figure}

The tail shape is fixed at zero, $\xi(w)\equiv0$, a consequence of the marginal standardization. On exponential margins $R\le d\max_j E_j$, so $P(R>r)\le d\,e^{-r/d}$ and heavy radial tails are excluded; for light-tailed sample clouds the radial decay rate in direction $w$ is fixed by the limit set \citep{nolde2022}, so the conditional radial tail is of exponential type. With $\xi=0$ the generalized Pareto distribution is threshold stable, so the tail scale keeps its meaning as the threshold curve moves across its posterior, and the results of Section \ref{sec:theory} hold in closed form. The second-order curvature of the radial tail, which fixed-threshold methods absorb with a free shape parameter \citep{mackay2023,wadsworthcampbell2024}, is taken up here by the estimated threshold and the Gamma body. Under independence the radial survivor $(1+r)e^{-r}$ carries a polynomial factor that the splice replaces by a constant; by Proposition \ref{prop:slowvar} the effect is confined to slowly varying factors, leaving decay rates unchanged.

\subsection{The angular density}\label{sec:angular}

The angular density is also a Bernstein polynomial \citep{hanson2017},
\[
f_W(w)=\sum_{k=1}^{K_W}\phi_k\,\mathrm{be}(w;k,K_W-k+1),
\]
a mixture of Beta densities $\mathrm{be}(\cdot;a,b)$ with weights $\bphi$ in the simplex $\mathcal{S}_{K_W-1}$ and $K_W$ fixed at a moderate value; for $d>2$ the basis consists of Dirichlet densities with integer parameters, $\tbinom{K_W+d-2}{d-1}$ of them for the same degree $K_W-1$, $55$ at $d=3$ and $715$ at $d=5$ for $K_W=10$, so the angular density is the easy part of the extension in principle and the expensive one in practice. Since $W$ is the angle of the full sample, the weights are free of the moment constraint that spectral measures of extreme value distributions must satisfy. Point masses at the vertices are not needed either: on exponential margins the dependence class is carried by the tail scale, and under independence $W$ is uniform.

\section{Extremal dependence properties}\label{sec:theory}

\subsection{The implied limit set}\label{sec:limitset}

Since the largest of $n$ standard exponential observations is of size $\log n$, the scaled cloud of Section \ref{sec:parfun} stays within the unit square and touches each of its sides, which the endpoint conditions $\sigma(0)=\sigma(1)=1$ ensure at $(1,0)$ and $(0,1)$. Figure \ref{fig:geometry} displays scaled samples for three dependence structures with their limit sets and the coefficient geometry of Section \ref{sec:coefficients}; the same display with the posterior of $\sigma(\cdot)$ over the scaled data serves as a check of fit in Section \ref{sec:predictive}.

\begin{figure}[t]
\centering
\input{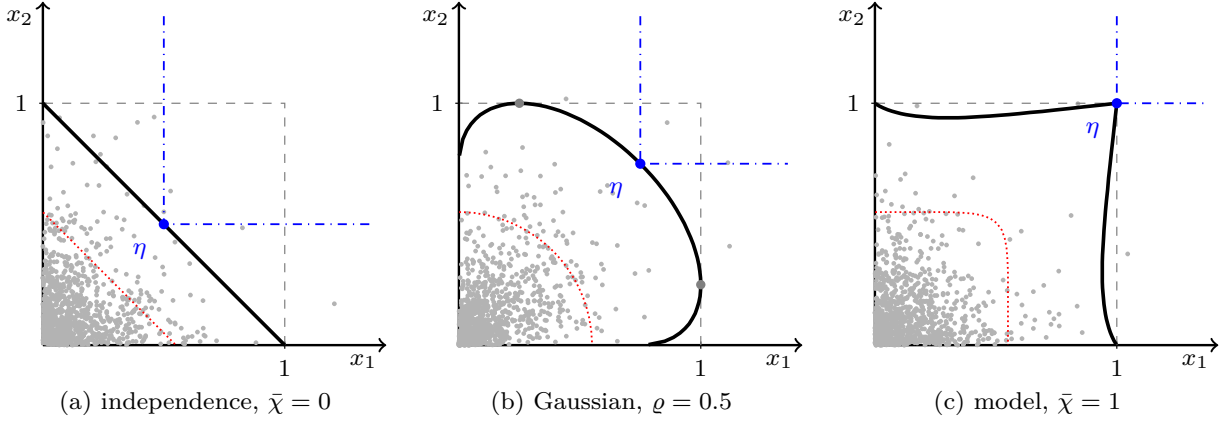}%
\caption{The geometry of extremal dependence for three reference cases: independence (left), a Gaussian copula with $\varrho=0.5$ (centre), and a sample from the model at the asymptotic dependence boundary $\bar{\chi}=1$ with $\nu\equiv0.25$ (right). Each panel shows $n=1000$ observations on standard exponential margins divided by $\log n$ (grey points), the boundary of the limit set at radius $\sigma(w)$ in direction $(w,1-w)$ (solid), the unit marginal box (dashed), the square $[\eta,\infty)^2$ meeting the limit set at $(\eta,\eta)$ (dash-dotted), with $\eta$ equal to $1/2$, $(1+\varrho)/2=3/4$ and $1$ respectively, and an illustrative threshold curve from the power-mean family \eqref{eq:rhop} with exponent $p$ equal to $1$, $2$ and $8$ (dotted). The scaled cloud fills its limit set, occasional points falling marginally outside at finite $n$; the Gaussian silhouette is pinched at the axes, at height $1-\varrho^2$, and meets the marginal box at the two interior points marked in grey (Section \ref{sec:coefficients}).}
\label{fig:geometry}
\end{figure}

The gauge function describes the limit set analytically. For a compact star-shaped set $G$ containing the origin, the gauge $g(\bx)=\inf\{s>0:\bx/s\in G\}$ is positive and homogeneous of order one, with $G=\{\bx\ge\bm{0}:g(\bx)\le1\}$, so that in direction $(w,1-w)$ the boundary of $G$ lies at radius $1/g(w,1-w)$. For light-tailed vectors, $-\log f_{\bE}(t\bx)/t\to g(\bx)$ as $t\to\infty$ \citep{nolde2022}, so the gauge is the direction-dependent exponential decay rate of the joint density. In radial models whose conditional radial tail decays exponentially this rate is the gauge \citep{papastathopoulos2023,wadsworthcampbell2024}, and Proposition \ref{prop:gauge} states the same for the proposed model, whose body and bounded threshold curve leave the limit set unaffected.

\begin{proposition}\label{prop:gauge}
Let $\bE$ have density \eqref{eq:factorization} with radial component \eqref{eq:splice}, $\xi\equiv0$, angular density $f_W$ positive on $(0,1)$, threshold function $\rho(\cdot)$ bounded, and tail scale $\sigma(\cdot)$ positive and continuous on $[0,1]$. Then, for every $\bx$ in the interior of the positive quadrant,
\[
-\frac{1}{t}\log f_{\bE}(t\bx)\xrightarrow[t\to\infty]{}g(\bx)=\frac{\normone{\bx}}{\sigma(x_1/\normone{\bx})},
\]
and the scaled sample cloud converges onto the limit set $G=\{\bx\ge\bm{0}:g(\bx)\le1\}$.
\end{proposition}

The convergence is in the sense of \citet{nolde2014}, and the proof, which needs only the exponential form of the radial survivor above the threshold, is in Section~\ref{app:proofs}, with the proofs of the other results of this section.

By Proposition \ref{prop:gauge}, the tail scale is the reciprocal gauge along directions, so the fitted $\sigma(\cdot)$ describes the joint tail geometry completely. Since any gauge in exponential margins satisfies $g(\bx)\ge\max_j x_j$, a valid limit set requires the envelope
\begin{equation}\label{eq:envelope}
\sigma(w)\le\min\Big\{\frac{1}{w},\frac{1}{1-w}\Big\}, \qquad w\in(0,1),
\end{equation}
and the limit set induced by the envelope is the unit box, dashed in the figures. Under \eqref{eq:scalefam} the envelope holds for every parameter value, so every posterior draw implies a valid limit set, the guarantee that \citet{majumder2025} obtain through constrained priors.

\subsection{Dependence coefficients}\label{sec:coefficients}

Standard summaries of extremal dependence follow directly. The residual tail dependence coefficient $\eta\in(0,1]$ of \citet{ledford1996} governs the decay of joint survival probabilities and separates asymptotic dependence ($\eta=1$ with slowly varying component tending to a positive constant, giving $\chi>0$) from asymptotic independence ($\eta<1$, giving $\chi=0$), where $\chi$ is the tail dependence coefficient and $\bar{\chi}=2\eta-1$ \citep{coles1999}. In exponential margins the coefficient is determined by the limit set through a quadrant correspondence: $\eta^{-1}=\inf\{g(\bx):\bx\ge(1,1)\}$, equivalently $\eta$ is the largest $s$ for which the square $[s,\infty)^2$ intersects $G$ \citep{nolde2014,nolde2022}; the often-quoted reduction $\eta=1/g(1,1)$ holds only when the infimum is attained at the corner. Proposition \ref{prop:eta} writes it along directions and evaluates it under the family \eqref{eq:scalefam}.

\begin{proposition}\label{prop:eta}
Under the conditions of Proposition \ref{prop:gauge},
\[
\eta=\sup_{w\in(0,1)}\,\min(w,1-w)\,\sigma(w).
\]
Under \eqref{eq:envelope}, $\eta=1$, which $\chi>0$ requires, holds if and only if $\sigma(1/2)=2$, when the limit set contains the corner $(1,1)$. Under the scale family \eqref{eq:scalefam}, the supremum is attained only at $w=1/2$, for every $\bar{\chi}\in[0,1]$ and every positive $\nu(\cdot)$, so that $\eta=(1+\bar{\chi})/2$.
\end{proposition}

The model parameter $\bar{\chi}=2\eta-1$ is therefore the coefficient of extremal association, with the asymptotic dependence boundary at $\bar{\chi}=1$, and its exact posterior distribution, including point masses on the two dependence classes under the priors of Section \ref{sec:classpriors}, is available from any posterior sample at no additional cost.

Table \ref{tab:copulas} reads familiar dependence structures through the model. Independence gives $\sigma\equiv1$ and $\bar{\chi}=0$, and a Gaussian copula with correlation $\varrho\ge0$ has diagonal value $\sigma(1/2)=1+\varrho$, recovering the classical value $\bar{\chi}=\varrho$ \citep{coles1999}. Asymptotically dependent structures such as the logistic attain the envelope value $\sigma(1/2)=2$, where the limit set touches the corner $(1,1)$; this geometric distinction between the two dependence classes is the one studied by \citet{lee2025geometric}. The Frank copula, associated in the body yet independent in the tail, separates body and tail dependence, which the model is built to express.

\begin{table}[t]
\centering
\caption{Familiar dependence structures read through the model: coefficients $\eta$, $\bar{\chi}$ and $\chi$, the scale function $\sigma(w)$ of the structure where available, and its position relative to the scale family \eqref{eq:scalefam} and the constraints \eqref{eq:exactclass}. Within the model $\chi(q)\to0$ on the exponential scale for every $\bar{\chi}$, the copula being asymptotically dependent only at the atom $\bar{\chi}=1$ (Proposition \ref{prop:slowvar}); negative association is deferred to Section \ref{sec:discussion}.}
\label{tab:copulas}
\footnotesize
\setlength{\tabcolsep}{4pt}
\resizebox{\textwidth}{!}{%
\begin{tabular}{lcccll}
\toprule
Structure & $\eta$ & $\bar{\chi}$ & $\chi$ & $\sigma(w)$ & Model \\
\midrule
Independence & $\frac12$ & $0$ & $0$ & $1$ & in \eqref{eq:scalefam}, $\bar{\chi}=0$; uniform angles, Gamma$(2,1)$ body \\[2pt]
Frank, $\theta$ & $\frac12$ & $0$ & $0$ & $1$ & in \eqref{eq:scalefam}, $\bar{\chi}=0$; body association via $f_W$, $\lambda(\cdot)$ \\[2pt]
Gaussian, $\varrho\in[0,1)$ & $\frac{1+\varrho}{2}$ & $\varrho$ & $0$ & $\frac{1-\varrho^2}{1-2\varrho\sqrt{w(1-w)}}$ & $\eta=\sigma(1/2)/2$, $\bar{\chi}=\varrho$; neither \eqref{eq:scalefam} nor \eqref{eq:exactclass} \\[6pt]
Inverted logistic, $\alpha_{\mathrm{IL}}$ & $2^{-\alpha_{\mathrm{IL}}}$ & $2^{1-\alpha_{\mathrm{IL}}}-1$ & $0$ & $\{w^{1/\alpha_{\mathrm{IL}}}+(1-w)^{1/\alpha_{\mathrm{IL}}}\}^{-\alpha_{\mathrm{IL}}}$ & $\eta=\sigma(1/2)/2$; satisfies \eqref{eq:exactclass}, not in \eqref{eq:scalefam} \\[4pt]
Logistic, $\alpha_{\mathrm{L}}\in(0,1)$ & $1$ & $1$ & $2-2^{\alpha_{\mathrm{L}}}$ & limit set touches $(1,1)$ & atom $\bar{\chi}=1$ \\[2pt]
Student $t$, $\varrho>0$ & $1$ & $1$ & $\in(0,1)$ & limit set touches $(1,1)$ & atom $\bar{\chi}=1$ \\[2pt]
Comonotonicity & $1$ & $1$ & $1$ & diagonal segment to $(1,1)$ & degenerate angle; outside the framework \\
\bottomrule
\end{tabular}}
\end{table}

The last column of Table \ref{tab:copulas} places each structure relative to the scale family and to the scale functions with $\sigma(1/2)=1+\bar{\chi}$ whose supremum in Proposition \ref{prop:eta} is attained at the diagonal, so that $\bar{\chi}$ is still exactly the coefficient of extremal association. These are the scale functions satisfying
\begin{equation}\label{eq:exactclass}
\sigma(w)\le\min\left\{1+\tau(w),\ \frac{1+\bar{\chi}}{2\min(w,1-w)}\right\},\qquad \sigma(0)=\sigma(1)=1.
\end{equation}
The first bound is the envelope \eqref{eq:envelope}, since $1+\tau(w)=1/\max(w,1-w)$, and the second keeps $\min(w,1-w)\,\sigma(w)$ below its diagonal value $(1+\bar{\chi})/2$. Since $\tau^{1+\nu}\le\tau$, the scale family lies below the cap $1+\bar{\chi}\tau(w)$, which satisfies both bounds, so the family satisfies \eqref{eq:exactclass}. The inverted logistic satisfies \eqref{eq:exactclass} but is not in the family, and the Gaussian is in neither, its scale exceeding the cap at intermediate angles and equal to $1-\varrho^2$ at the axes, where its limit set is pinched (Figure \ref{fig:geometry}, centre). Under such copulas the scale family inflates $\bar{\chi}$ while finite-level summaries remain accurate (Section \ref{sec:s3}), so the breadth of \eqref{eq:exactclass} matters for asymptotic summaries only.

\subsection{Joint tail decay}\label{sec:jointdecay}

The limit set fixes tail decay only on the logarithmic scale; the first result gives the exact decay of the joint and marginal survivors under the scale family.

\begin{proposition}\label{prop:slowvar}
Assume the conditions of Proposition \ref{prop:gauge} with the scale family \eqref{eq:scalefam}, and let $f_W$, $\rho(\cdot)$, $\lambda(\cdot)$ and $\nu(\cdot)$ be continuous and $f_W$ positive on $[0,1]$. Then there are positive constants $c$, $c_1$ and $c_2$ such that, as $t\to\infty$,
\begin{equation}\label{eq:slowvar}
P(E_1>t,E_2>t)=\frac{c}{t}\,e^{-t/\eta}\{1+o(1)\},\qquad P(E_j>t)=\frac{c_j}{t}\,e^{-t}\{1+o(1)\}\quad(j=1,2).
\end{equation}
Consequently $\chi(q)=P(E_1>t_q,E_2>t_q)/(1-q)$, with $t_q=-\log(1-q)$, tends to zero for every $\bar{\chi}\in[0,1]$, while the copula of the model is asymptotically independent with coefficient $\eta$ if $\bar{\chi}<1$ and asymptotically dependent if $\bar{\chi}=1$.
\end{proposition}

Both survivors carry a factor $t^{-1}$, the joint one from the kink of $\min(w,1-w)\,\sigma(w)$ at the diagonal and the marginal ones from the axes. On the scale of the data, whose margins are exactly exponential, only the joint factor counts, and at the boundary $\chi(q)$ vanishes logarithmically; relative to the margins of the model the two factors cancel there, and the copula is asymptotically dependent. The marginal factor also makes the margins implied by a fit lighter than the standardized ones in the far tail, which is why the return level sets of Section \ref{sec:application} are tight near the axes. The constants, given in Section~\ref{supp:proof-const}, depend on the relaxation function $\nu(\cdot)$ only through $\nu(1/2)$, smaller values giving stronger sub-asymptotic dependence.

The second result is the angular dependence function of \citet{wadsworthtawn2013}, which \citet{nolde2022} express through the gauge of the limit set; for the model it takes an explicit form, with $\eta$ as the diagonal case $\omega=1/2$.

\begin{proposition}\label{prop:adf}
Under the conditions of Proposition \ref{prop:gauge}, for every $\omega\in(0,1)$,
\begin{equation}\label{eq:adf}
-\frac{1}{t}\log P(E_1>\omega t,\,E_2>(1-\omega)t)\xrightarrow[t\to\infty]{}\Lambda(\omega)=\inf_{w\in(0,1)}\frac{1}{\sigma(w)}\max\left(\frac{\omega}{w},\frac{1-\omega}{1-w}\right),
\end{equation}
the angular dependence function, and $\Lambda(1/2)=1/(2\eta)$. Under \eqref{eq:scalefam}, $\max(\omega,1-\omega)\le\Lambda(\omega)\le1$ for every $\omega$, with $\Lambda(\omega)$ pointwise nonincreasing in $\bar{\chi}$, $\Lambda\equiv1$ if and only if $\bar{\chi}=0$, and $\Lambda(\omega)=\max(\omega,1-\omega)$ for every $\omega$ if and only if $\bar{\chi}=1$.
\end{proposition}

The two boundary curves of the angular dependence function correspond to the two values of $\bar{\chi}$ at which the prior of Section \ref{sec:classpriors} places point masses: $\Lambda\equiv1$ is tail independence ($\bar{\chi}=0$), ${\Lambda(\omega)=\max(\omega,1-\omega)}$ is the asymptotic dependence boundary ($\bar{\chi}=1$), and intermediate values of $\bar{\chi}$ trace the intermediate curves, monotonically.

\subsection{Extremal dependence in $d$ dimensions}\label{sec:multi}

The general statements of Propositions \ref{prop:gauge} to \ref{prop:adf} hold in $d$ dimensions, with the angle $\bw$ ranging over $\mathcal{S}_{d-1}$ and the envelope \eqref{eq:envelope} becoming $\sigma(\bw)\le1/\max_jw_j$; only the polynomial factor $t^{-1}$ of Proposition \ref{prop:slowvar} changes, becoming $t^{-(d-1)}$ (Section~\ref{supp:multi-theory}). In particular, the residual tail dependence coefficient of any subset $C$ of at least two components is
\[
\eta_C=\sup_{\bw\in\mathcal{S}_{d-1}}\,\min_{j\in C}w_j\,\sigma(\bw),
\]
so one fit returns the coefficients of all $2^d-d-1$ subsets, the full vector being asymptotically dependent if and only if $\sigma$ reaches its envelope value $d$ at the barycentre. With the factor $t^{-(d-1)}$, $\chi(q)$ vanishes on the data scale like $\{-\log(1-q)\}^{-(d-1)}$ when $\eta=1$, and the marginal mismatch noted after Proposition \ref{prop:slowvar} grows with $d$.

The scale family \eqref{eq:scalefam} does not extend directly: replacing $\tau$ by $\min_jw_j/\max_jw_j$ gives a scale equal to one on every face of the simplex, so every two-dimensional face of the limit set is that of independence and a dependent pair with an independent third component cannot be represented. Instead, we write \eqref{eq:scalefam} in terms of its gauge, $g(\bx)=\normone{\bx}-D(\bx)$, where $\normone{\bx}$ is the gauge of independence and the deficit $D(\bx)=\normone{\bx}\{1-1/\sigma(x_1/\normone{\bx})\}$ is homogeneous of order one, nonnegative and zero on the axes, and subtract one deficit per pair,
\begin{equation}\label{eq:addgauge}
g(\bx)=\normone{\bx}-\sum_{j<k}D_{jk}(x_j,x_k),\qquad D_{jk}(a,b)=(a+b)\Big[1-\frac{1}{\sigma_{jk}\{a/(a+b)\}}\Big],
\end{equation}
with $\sigma_{jk}$ of the form \eqref{eq:scalefam} with parameters $(\bar{\chi}_{jk},\nu_{jk})$, and $\sigma(\bw)=1/g(\bw)$ on the simplex.

The family has three properties, illustrated in Figure \ref{fig:geometry3}. First, all deficits but $D_{jk}$ vanish on the face where only $x_j$ and $x_k$ are nonzero, so each two-dimensional face of the limit set is a limit set of \eqref{eq:scalefam}. Second, the envelope couples the pairs: at $d=3$ with equal $\bar{\chi}_{jk}$ and $\nu_{jk}=1/2$ it holds up to $\bar{\chi}_{jk}=1/2$, where the limit set reaches the corner $(1,1,1)$, while a single pair may reach $\bar{\chi}_{jk}=1$ when the others are zero, so admissibility has to be checked at every proposal. Third, $\eta_{jk}$ equals its bivariate value $(1+\bar{\chi}_{jk})/2$ only when the other pairs are weak: with $(\bar{\chi}_{12},\bar{\chi}_{13},\bar{\chi}_{23})=(0.6,0.2,0.2)$, $\eta_{12}=0.80$ is the bivariate value, but with $(0.5,0.4,0.4)$, $\eta_{12}=0.84$ against $0.75$, because the deficits add in the interior of the simplex and the projection of the limit set extends beyond its face. The coefficients are therefore no longer read off a parameter as in Proposition \ref{prop:eta}, and class priors in $d$ dimensions, on the $\bar{\chi}_{jk}$ or on the coefficients directly \citep{simpson2020}, are left for future work.

\begin{figure}[t]
\centering
\includegraphics[width=\textwidth]{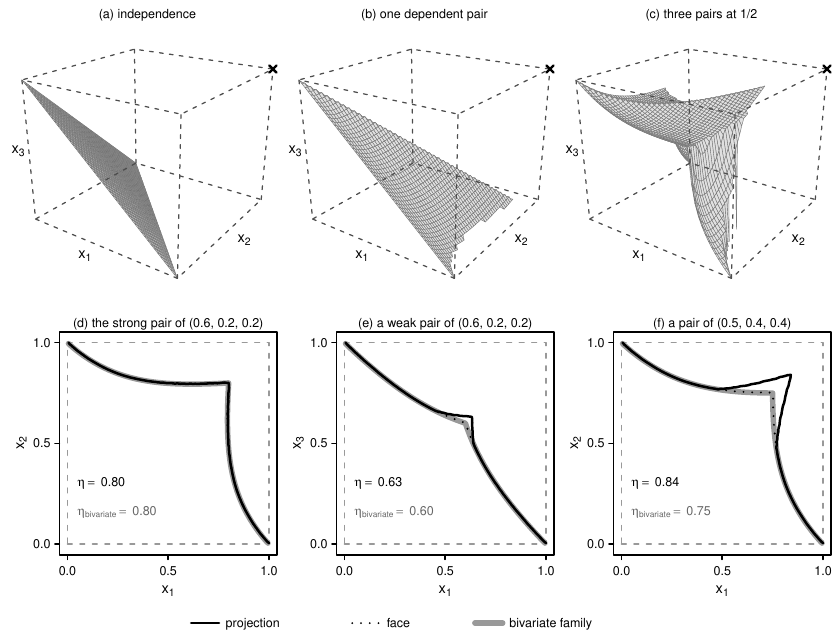}
\caption{Trivariate limit sets of \eqref{eq:addgauge} with $\nu_{jk}=1/2$. Top: the limit set boundary in the unit cube (dashed) for (a) independence, (b) one dependent pair, $\bar{\chi}=(0.8,0,0)$, and (c) three pairs at $\bar{\chi}_{jk}=1/2$, where the corner $(1,1,1)$ (cross) is attained. Bottom: projection of the limit set onto a coordinate plane, the corresponding face, and the bivariate limit set of \eqref{eq:scalefam} with the same $\bar{\chi}_{jk}$, with the coefficient $\eta$ of the projection and of the bivariate family. The projection coincides with the face for (d) the strong pair of $\bar{\chi}=(0.6,0.2,0.2)$, extends beyond it for (e) a weak pair, and extends further for (f) a pair of $\bar{\chi}=(0.5,0.4,0.4)$, through the conical maximum of $\sigma$ at the barycentre.}
\label{fig:geometry3}
\end{figure}

\section{Bayesian inference}\label{sec:inference}

STAR describes dependence on the standardized scale, within a joint specification of margins and dependence. Each margin is modelled by a distribution $F_j(\cdot\,;\bpsi_j)$ with density $f_j$, the empirical distribution function below a marginal threshold and a GPD above it \citep{coles2001}, or a univariate spliced model with estimated threshold \citep{behrens2004,nascimento2012}, and $T_j(y;\bpsi_j)=-\log\{1-F_j(y;\bpsi_j)\}$ maps it to an exponential margin. With $T$ acting componentwise, $\bpsi=(\bpsi_1,\ldots,\bpsi_d)$ and $\btheta$ the dependence parameters of Section \ref{sec:likelihood}, the joint density of an observation $\bm{y}$ on its original scale is
\[
p(\bm{y}\mid\bpsi,\btheta)=f_{\bE}\{T(\bm{y};\bpsi);\btheta\}\prod_{j=1}^d\frac{f_j(y_j;\bpsi_j)}{1-F_j(y_j;\bpsi_j)},
\]
the Jacobian of $T$ multiplying the model density \eqref{eq:factorization}, its factors below the marginal threshold free of $\bpsi_j$. The two blocks share no parameter and are variation independent, so under independent priors the joint posterior of $(\bpsi,\btheta)$ can be sampled by alternating the scheme of Section \ref{sec:mcmc} for $\btheta$ given $\bpsi$ with Metropolis updates of $\bpsi$ given $\btheta$, at the cost of re-standardizing the sample and re-evaluating \eqref{eq:loglik} at every marginal move.

The contribution of this paper is the dependence construction, and we therefore fix $\bpsi$ at its estimate, standardize once, and treat the resulting $\bm{e}_i=T(\bm{y}_i;\hat{\bpsi})$ as data. This two-step approach, the plug-in cut, is a cut of the joint posterior with the marginal block collapsed to a point estimate \citep{plummer2015}, which keeps a working model for the tail from disturbing marginal fits obtained from all of the data. For summaries on the original scale, such as return level sets, we use the cut posterior itself, drawing $\bpsi$ from its marginal posterior and refitting $\btheta$. In simulations (Section~\ref{supp:cutjoint}) the three posteriors agree on $\chi(q)$; on the original scale the cut posterior recovers most of the coverage that the plug-in cut loses, and the joint posterior more, through feedback towards the lighter margins of the model (Proposition \ref{prop:slowvar}), which also shifts the class probabilities.

\subsection{Likelihood}\label{sec:likelihood}

Let $\bm{e}_1,\ldots,\bm{e}_n$ denote the standardized observations with polar coordinates $(r_i,w_i)$, and collect the unknowns as $\btheta=(\bphi,\alpha,\bm{b},\btheta_\rho,\bar{\chi},\bm{z})$, where $\btheta_\rho=(u_1,u_2,p)$ collects the parameters of the threshold family \eqref{eq:rhop}. Writing $B(\btheta_\rho)=\{i:r_i\le\rho(w_i)\}$ and $T(\btheta_\rho)=\{i:r_i>\rho(w_i)\}$ for the body and tail index sets, the log-likelihood is
\begin{align}
\ell(\btheta)={}&\sum_{i=1}^n\Big[\log f_W(w_i;\bphi)-(d-1)\log r_i\Big]
+\sum_{i\in B(\btheta_\rho)}\log g_{\mathrm{b}}(r_i\mid w_i)\nonumber\\
&+\sum_{i\in T(\btheta_\rho)}\Big[\log\{1-G_{\mathrm{b}}(\rho(w_i)\mid w_i)\}+\log h_{\mathrm{GP}}\{r_i-\rho(w_i);\sigma(w_i),0\}\Big],\label{eq:loglik}
\end{align}
which is exact and available in closed form: the membership of each observation is determined by its own angle and radius, no latent allocation variables are required, and each evaluation costs $O(n)$ operations. As the threshold parameters move, observations cross the boundary and the likelihood is discontinuous in $\btheta_\rho$; this poses no difficulty for the Metropolis-Hastings schemes below, which require only pointwise evaluation.

\subsection{Prior specification}\label{sec:priors}

Priors are weakly informative, with two purposes: to anchor the threshold curve in the upper region of the data and to regularize the functional coefficients. The angular weights receive a $\mathrm{Dirichlet}(a,\ldots,a)$ prior with $a=1$; the body parameters, $\log\alpha\sim N(\log d,v_\alpha^2)$, centred at independence, and $b_k\sim N(0,v_\lambda^2)$; and the tail coefficients, $z_k\sim N(0,v_\nu^2)$. The anchor levels follow the location-aware construction of \citet{hu2026}, $u_j=q_{\kappa_0}+o_j$, with $q_{\kappa_0}=-\log(1-\kappa_0)$ the standard exponential quantile at a reference level such as $\kappa_0=0.8$ and independent half-normal offsets $o_j$. The exponent receives $\log p\sim N(\log 2,v_p^2)$, which also supports concave boundaries. Default hyperparameters are listed in Section~\ref{app:computation}. Over five paired datasets at $n=1000$ (Section~\ref{supp:sens}), moving $\kappa_0$ to $0.7$ or $0.9$, or the half-normal scale to one half or two, shifts on average the exceedance fraction by at most $0.025$ and the diagonal threshold by at most $0.15$, but $\bar{\chi}$ by at most $0.02$, against an interval width near $0.61$, and $\chi(0.99)$ by under four percent: the prior locates the threshold without driving the dependence summaries.

\subsection{Priors for the dependence classes}\label{sec:classpriors}

By Proposition \ref{prop:eta}, asymptotic dependence sits at the boundary $\bar{\chi}=1$ of the parameter space, so the prior on $\bar{\chi}$ decides how the dependence classes are handled. A continuous prior, $\bar{\chi}\sim\mathrm{Beta}(a_\chi,b_\chi)$ with the uniform as default, reads the class from the posteriors of $\bar{\chi}$, $\eta$, $\chi(q)$ and $\bar{\chi}(q)$, acknowledging that at finite levels weak asymptotic dependence and asymptotic independence are hard to tell apart \citep{hu2026}. Class probabilities require instead a three-component prior, in the tradition of spike-and-slab priors \citep{mitchell1988},
\[
P(\bar{\chi}=0)=\pi_0, \qquad P(\bar{\chi}=1)=\pi_1, \qquad \bar{\chi}\mid\bar{\chi}\in(0,1)\sim\mathrm{Beta}(a_\chi,b_\chi),
\]
with $\pi_0=\pi_1=1/4$ and a uniform slab as default. The atom at zero is tail independence, with $\sigma\equiv1$ and the unit simplex as limit set; the atom at one is the asymptotic dependence boundary, where the limit set touches the corner $(1,1)$, the geometric criterion for asymptotic dependence \citep{lee2025geometric}; and the slab measures the strength of extremal association within asymptotic independence. Both atoms are classes, with the angular density and the body left free: the first collects every structure with the joint tail of independence, the second every asymptotically dependent structure, the fitted $\chi(q)$ on the scale of the data still vanishing logarithmically (Proposition \ref{prop:slowvar}). The atoms are what make the classification possible, since both classes have probability zero under any continuous prior on $\bar{\chi}$. The class probabilities depend much more on the atom weights than on the slab (Section~\ref{supp:sens}), and are reported together with the prior weights that produced them.

\subsection{Posterior computation}\label{sec:mcmc}

Posterior sampling proceeds by adaptive Metropolis-within-Gibbs \citep{roberts2009}, cycling through blocks for the angular weights, the body parameters, the threshold coordinates and the tail parameters on unconstrained scales, with the between-component moves of the three-component prior handled by a jump step that requires no reversible-jump correction because the slab proposal coincides with the slab prior. Every quantity of Section \ref{sec:theory} is a deterministic function of the parameters, so posterior draws of the coefficients, the limit set and the dependence curves are obtained by evaluation. Details, including proposal scales, initialization and acceptance rates, are in Section~\ref{app:computation}.

The blocks interact through a trade-off: a higher threshold with a heavier body can imitate a lower threshold with a lighter tail over the observed range. Three devices control it: the anchoring priors of Section \ref{sec:priors} keep the threshold curve in the joint upper region of the data; the structured family \eqref{eq:scalefam} pins the diagonal value of the scale at $1+\bar{\chi}$; and the exact threshold stability at $\xi=0$ prevents a move of the curve from being compensated by rescaling the tail. In practice the trade-off is largely inactive, posterior correlations across the blocks remaining modest even under a deliberately weak prior (Section~\ref{supp:ident}). Two diagnostics signal joint tail mass that the $\xi=0$ splice cannot carry, as with genuinely positive $\chi$: a threshold posterior piling up against its prior floor, with an exceedance fraction far above the reference level, and a fitted $\chi(q)$ that decays, as Proposition \ref{prop:slowvar} requires, while the empirical $\chi(q)$ stays flat. The second, a posterior predictive check, can fire when the first does not, as in the coastal application of Section~\ref{supp:duck}; the remedy is an asymptotically dependent construction of the kind discussed in Section \ref{sec:discussion}.

\subsection{Posterior predictive analysis and derived quantities}\label{sec:predictive}

The model is generative, so posterior predictive simulation is immediate: draw an angle from $f_W$, a body indicator with probability $\varpi(w)$, then a truncated Gamma radius on $(0,\rho(w))$ or a GPD excess above $\rho(w)$, and map back to the original scale. Joint tail probabilities need no simulation, since $E_1>t_1$ and $E_2>t_2$ hold exactly when ${R>\max\{t_1/W,t_2/(1-W)\}}$:
\begin{equation}\label{eq:jointsurv}
P(E_1>t_1,E_2>t_2)=\int_0^1 f_W(w)\,\bar{F}_{R\mid W}\Big(\max\Big\{\frac{t_1}{w},\frac{t_2}{1-w}\Big\}\,\Big|\,w\Big)\,dw,
\end{equation}
where the conditional radial survivor $\bar{F}_{R\mid W}(r\mid w)$ is $1-G_{\mathrm{b}}(r\mid w)$ for $r\le\rho(w)$ and $\{1-\varpi(w)\}\exp\{-(r-\rho(w))/\sigma(w)\}$ above it, so that \eqref{eq:jointsurv} is a one-dimensional quadrature of a closed-form integrand. With $t_1=t_2=t_q=-\log(1-q)$ and $p_q$ the resulting probability, it yields the sub-asymptotic curves $\chi(q)=p_q/(1-q)$ and $\bar{\chi}(q)=2\log(1-q)/\log p_q-1$ \citep{coles1999}, the model-based counterparts of the usual empirical diagnostics. From the posterior draws we report the threshold curve and the exceedance fraction; the coefficients $\bar{\chi}$ and $\eta$ with the class probabilities; the curves $\chi(q)$, $\bar{\chi}(q)$ and $\Lambda(\cdot)$ with pointwise bands; the angular density of exceedances, proportional to $f_W(w)\{1-\varpi(w)\}$; joint exceedance probabilities and quantiles of applied interest; and return level sets and environmental contours \citep{simpsontawn2025} derived from the fitted boundary.

Three of these outputs are exactly what fixed-threshold tail methods cannot return. Every posterior draw of $\sigma(\cdot)$ defines a valid limit set, so the posterior of the limit set boundary can be overlaid on the data divided by $\log n$, as in Figure \ref{fig:geometry}, as a check of fit that discards no draws. The threshold curve has a posterior, so the exceedance region, the exceedance fraction and the effective number of tail observations are random quantities whose uncertainty propagates into every tail summary. Because the model has full support with exact normalization, the probability of a region lying partly below and partly above the threshold curve, which combines moderate and extreme values, is computed in the same way as that of a purely extremal region; the compound events of Section \ref{sec:intro} are of this kind.

\section{Simulation studies}\label{sec:simulations}\label{sec:identification}\label{sec:s2}\label{sec:s3}

The simulation study covers identification, calibration under the model, classification of the dependence class, misspecification, the value of the estimated threshold, and convergence and prior sensitivity. Data are generated from the model and, on standard exponential margins, from five copula families outside it, with the defaults of Section \ref{sec:priors}. Here we report, over fifty replicates, the comparison with competing estimators, the calibration, the classification and the threshold ablation; the other experiments, and the full results of all, are in Sections~\ref{app:computation} and~\ref{supp:sims}.

\begin{table}[t]
\centering
\caption{Root mean squared error of $\chi(q)$ at $n=4000$ for STAR, the conditional extremes model (HT), the geometric model (WC) and the empirical estimator (emp.), smallest in bold: within the model at $q=0.99$, HT on estimated margins (left), and under the Gaussian ($\varrho=0.5$) and inverted logistic ($\alpha_{\mathrm{IL}}=0.5$) copulas, HT on the exact margins (right). HT conditions on the $0.9$ quantile.}
\label{tab:key}
\footnotesize
\setlength{\tabcolsep}{3.5pt}
\resizebox{\textwidth}{!}{%
\begin{tabular}{c cccc c c cccc cccc}
\toprule
 & \multicolumn{4}{c}{$q=0.99$} & \hspace{0.5em} & & \multicolumn{4}{c}{Gaussian} & \multicolumn{4}{c}{Inverted logistic} \\
\cmidrule(lr){2-5}\cmidrule(lr){8-11}\cmidrule(lr){12-15}
$\bar{\chi}$ & STAR & HT & WC & emp. & & $q$ & STAR & HT & WC & emp. & STAR & HT & WC & emp. \\
\midrule
$0.15$ & $0.006$ & $\mathbf{0.005}$ & $0.008$ & $0.017$ & & $0.95$ & $0.024$ & $\mathbf{0.019}$ & $0.025$ & $0.036$ & $0.023$ & $\mathbf{0.022}$ & $0.032$ & $0.031$ \\
$0.45$ & $\mathbf{0.011}$ & $0.013$ & $0.013$ & $0.028$ & & $0.99$ & $\mathbf{0.025}$ & $0.026$ & $0.034$ & $0.057$ & $\mathbf{0.021}$ & $0.026$ & $0.035$ & $0.058$ \\
$0.75$ & $\mathbf{0.011}$ & $0.022$ & $0.020$ & $0.035$ & & $0.999$ & $0.031$ & $\mathbf{0.024}$ & $0.029$ & $0.104$ & $0.035$ & $\mathbf{0.024}$ & $0.032$ & $0.121$ \\
$0.95$ & $\mathbf{0.017}$ & $0.033$ & $0.030$ & $0.049$ & & & & & & & & & & \\
\bottomrule
\end{tabular}}
\end{table}

\begin{table}[t]
\centering
\caption{Left: coverage of $90\%$ intervals for $\chi(q)$ within the model at $n=4000$, HT on estimated margins. Centre: threshold ablation at $n=4000$, threshold estimated or frozen at the level-$\kappa$ quantile curve: posterior median of $\chi(0.999)$ under the model truth $\bar{\chi}=0.6$ and coverage under the STAR, Gaussian and Frank truths. Right: posterior probability of asymptotic dependence at $n=8000$ under the three-component prior.}
\label{tab:rest}
\footnotesize
\setlength{\tabcolsep}{3.5pt}
\resizebox{\textwidth}{!}{%
\begin{tabular}{l ccc c l c ccc c l c c}
\toprule
\multicolumn{4}{c}{coverage of $\chi(q)$} & & \multicolumn{5}{c}{threshold ablation, $\chi(0.999)$} & & \multicolumn{3}{c}{classification} \\
\cmidrule(lr){1-4}\cmidrule(lr){6-10}\cmidrule(lr){12-14}
$\bar{\chi}$, $q$ & STAR & HT & emp. & & threshold & median & STAR & Gauss. & Frank & & generator & class & $P(\mathrm{AD}\mid\text{data})$ \\
\midrule
$0.15$, $0.99$ & $0.96$ & $0.74$ & $0.30$ & & estimated & $0.021$ & $0.90$ & $0.78$ & $0.18$ & & STAR, $\bar{\chi}=0$ & AI & $0.00$ \\
$0.45$, $0.99$ & $0.92$ & $0.70$ & $0.68$ & & $\kappa=0.8$ & $0.021$ & $0.92$ & $0.10$ & $0.00$ & & STAR, $\bar{\chi}=0.5$ & AI & $0.06$ \\
$0.75$, $0.99$ & $0.98$ & $0.74$ & $0.90$ & & $\kappa=0.9$ & $0.017$ & $0.96$ & $0.48$ & $0.00$ & & STAR, $\bar{\chi}=0.95$ & AI & $0.74$ \\
$0.95$, $0.99$ & $0.78$ & $0.66$ & $0.86$ & & $\kappa=0.95$ & $0.013$ & $0.96$ & $0.86$ & $0.16$ & & STAR, $\bar{\chi}=1$ & AD & $0.81$ \\
$0.6$, $0.999$ & $0.90$ & $0.58$ & $0.04$ & & & & & & & & Gaussian, $\varrho=0.5$ & AI & $0.60$ \\
 & & & & & & & & & & & Gaussian, $\varrho=0.9$ & AI & $0.82$ \\
 & & & & & & & & & & & inverted logistic, $0.5$ & AI & $0.72$ \\
 & & & & & & & & & & & logistic, $0.5$ & AD & $0.77$ \\
 & & & & & & & & & & & Student $t$, $(0.5,4)$ & AD & $0.97$ \\
\bottomrule
\end{tabular}}
\end{table}

Table \ref{tab:key} compares STAR with the conditional extremes model of \citet{heffernan2004}, the geometric model of \citet{wadsworthcampbell2024} and the empirical estimator on the same data, and Table \ref{tab:rest} (left) gives the coverage of the corresponding intervals. Within the model, STAR is the most accurate estimator of $\chi(0.99)$ at all but the weakest dependence, and its intervals cover near the nominal rate away from the boundary, the bootstrap intervals of the conditional model in about seven replicates out of ten, and those of the empirical estimator in as few as three; at $q=0.999$ the coverages are $0.90$, $0.58$ and $0.04$. Outside the model, under the Gaussian and inverted logistic copulas, STAR is the most accurate at $q=0.99$ and at most $0.011$ behind the conditional model, given exact margins, at $q=0.95$ and $0.999$, the empirical estimator three to five times worse at $q=0.999$. The conditional model is sensitive to its threshold, whose best level is low under these copulas and high within the model at strong dependence (Section~\ref{supp:empirical}). Two copulas mark the limits predicted by the theory: under the Frank copula, associated in the body and independent in the tail, $\chi(q)$ is overestimated, and under the Student $t$ the $\xi=0$ splice underestimates positive $\chi$, which the diagnostic of Section \ref{sec:mcmc} flags. When the true scale function exceeds the cap of the scale family (Section \ref{sec:coefficients}), the model overestimates $\eta$, while the geometric model, which chooses its gauge, does not.

Estimating the threshold removes the sensitivity of far-tail summaries to an arbitrary fixed level (Table \ref{tab:rest}, centre). Frozen at the level-$\kappa$ quantile curve, the threshold moves the posterior median of $\chi(0.999)$ by almost forty percent between $\kappa=0.8$ and $0.95$ under the model truth, and under Gaussian data collapses its coverage to $0.10$ and $0.48$ at $\kappa=0.8$ and $0.9$ while the estimated threshold keeps $0.78$; the high level does as well, so the gain lies in not having to know which level is safe, and under the Frank copula no level does. STAR is regularly identified: body, threshold curve and exceedance fraction contract at about the parametric rate, and $\bar{\chi}$ more slowly, as Proposition \ref{prop:slowvar} implies. In the classification (Table \ref{tab:rest}, right), the asymptotically dependent truths receive posterior probability $0.77$ to $0.97$ on the atom at one, and the model at $\bar{\chi}=0$ and $0.5$ at most $0.06$. The strongly dependent asymptotically independent truths receive $0.60$ to $0.82$, with little change in $n$: the model at $\bar{\chi}=0.95$, whose joint tail differs from the boundary by a factor nearly constant over the observed range, and the Gaussian and inverted logistic copulas, for which the scale family also inflates $\bar{\chi}$ (Section \ref{sec:coefficients}). A small posterior probability of the atom at one is therefore reliable evidence of asymptotic independence, while a large one means that asymptotic dependence cannot be ruled out at the resolution of the data; class probabilities are reported with the prior weights.

\section{Application: co-pollution in Beijing}\label{sec:application}\label{sec:exploratory}\label{sec:results}

\subsection{Fine particulate matter and ozone}\label{sec:app2}

The application concerns air pollution in Beijing, using the multi-site data of \citet{zhang2017}: hourly concentrations at twelve monitoring stations from March 2013 to February 2017, a period over which the accompanying study documented a sustained fall in fine particulate matter alongside a rise in ozone, and which framed the resulting concern with co-pollution, days on which both are high. That concern is a question about a joint tail at mixed levels, precisely what STAR returns. We use the warm season, April to September, when ozone forms photochemically, and take daily maxima of PM$_{2.5}$ and ozone over the four seasons. Before computing the maxima, an hourly reading is set to missing as a sensor fault if it exceeds $300\,\mu$g/m$^3$ and is more than three times the largest reading at the other eleven stations in the same hour. We report Wanliu, and the same analysis at all twelve stations in Section~\ref{supp:stations}. At Wanliu $n=732$, the Spearman correlation is $0.23$, and the empirical $\chi(q)$ falls from $0.15$ at $q=0.9$ to $0.05$ at $q=0.95$ and to zero at $q=0.99$, the regime in which the classification of Section \ref{sec:classpriors} is informative. Each margin is standardized semiparametrically, the empirical distribution function below its $0.9$ quantile and a generalized Pareto distribution above it, with estimated shapes $0.14$ for particulate matter and $-0.20$ for ozone, a bounded tail.

Consecutive days are dependent, with lag-one autocorrelations of $0.50$ and $0.56$, but joint extremes are nearly isolated: the intervals estimator of the extremal index \citep{ferro2003} is $0.59$ for days beyond the estimated threshold curve and $0.94$ and $1$ for days on which both pollutants exceed their $0.8$ or $0.9$ quantiles. We fit all days without declustering, since the targets are probabilities per day and an estimated threshold offers no fixed level at which to decluster. A moving-block bootstrap with ten-day blocks, re-estimating the margins, gives standard deviations of the posterior means about $1.3$ times those from resampling single days, still below the posterior standard deviations of $\chi(q)$ and the joint probabilities. Following Section \ref{sec:threshold} the threshold exponent is restricted to convex curves, $p\ge1$, and the three-component prior of Section \ref{sec:classpriors} takes its defaults. Results are from a chain of $30{,}000$ iterations, with $10{,}000$ of burn-in and thinning by $5$; three more chains give potential scale reduction factors below $1.01$ and, over the four chains, effective sample sizes of at least $590$ for every parameter and of $1500$ to $4300$ for the exceedance fraction and the joint probabilities.

\begin{figure}[t]
\centering
\input{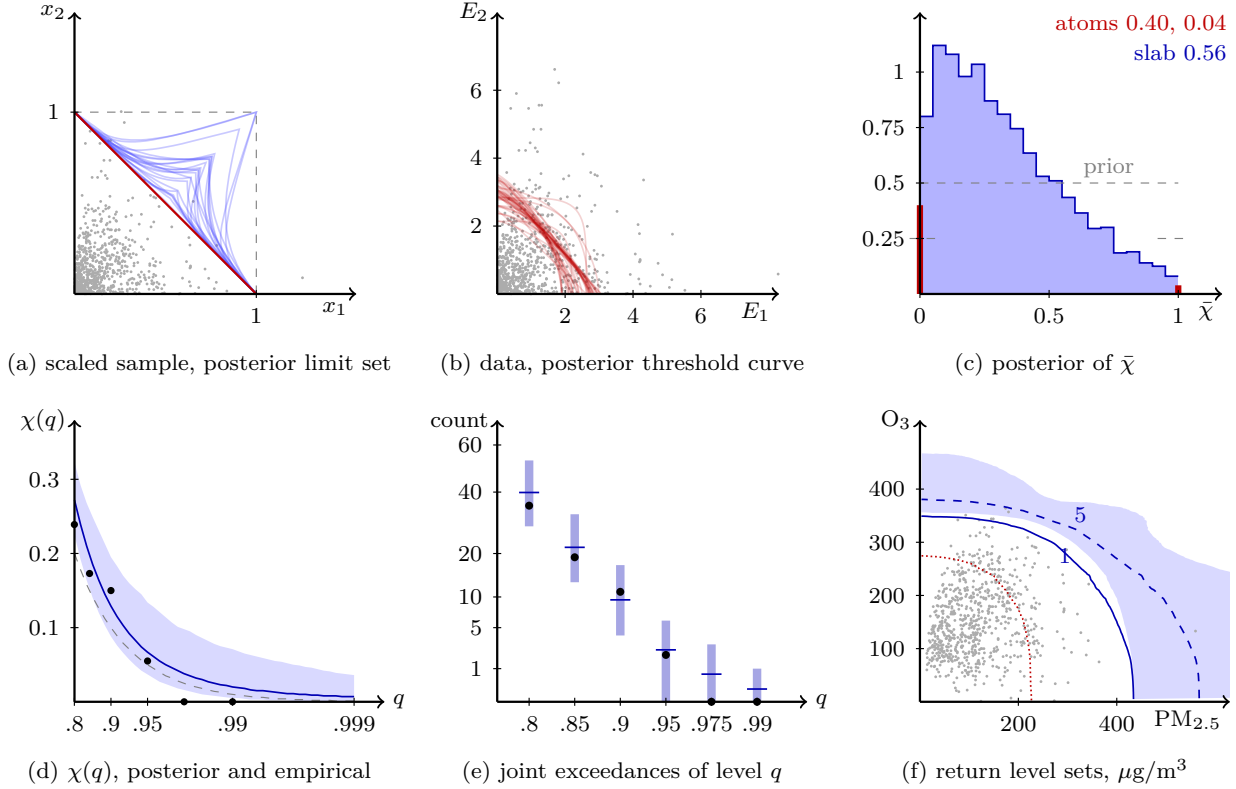}%
\caption{Beijing warm-season daily maxima of PM$_{2.5}$ ($E_1$) and ozone ($E_2$) at Wanliu. (a) Sample on standard exponential margins divided by $\log n$, with the limit set boundaries of forty posterior draws, those at the atom $\bar{\chi}=0$ sharing the line $x_1+x_2=1$ (red), and the unit box (dashed). (b) Data with forty posterior draws of the threshold curve. (c) Posterior of $\bar{\chi}$ under the three-component prior: atoms as bars with height equal to their probability, slab as its density times its probability, prior dashed. (d) $\chi(q)$ for $q$ from $0.8$ to $0.999$: posterior mean and $90\%$ band, empirical values (dots) and the independence value $1-q$ (dashed). (e) Number of days on which both pollutants exceed their level-$q$ quantiles: observed (dots) against the posterior predictive mean and $90\%$ interval (bars), on a square-root scale. (f) In $\mu$g/m$^3$, under the cut posterior: the data, the posterior median threshold curve (dotted), and the return level sets exceeded on average once per warm season (solid) and once in five (dashed, with $90\%$ band).}
\label{fig:beijing}
\end{figure}

Figure \ref{fig:beijing} collects the outputs of Section \ref{sec:predictive}, panels (a) and (b) showing the two objects the model estimates. The posterior draws of the limit set range from the triangle $x_1+x_2\le1$ of exact tail independence to a moderate bulge along the diagonal. The threshold curve meets the axes near the $0.93$ and $0.95$ marginal quantiles, with anchor posterior medians $2.7$ and $3.0$, and is nearly straight, with exponent posterior median $1.3$ and $90\%$ interval $(1.0,3.9)$, so a day enters the extremal region roughly when the combined standardized load $E_1/u_1+E_2/u_2$ exceeds one; on the original scale the curve passes through $171\,\mu$g/m$^3$ of PM$_{2.5}$ and $212\,\mu$g/m$^3$ of ozone on the diagonal (panel (f), dotted). The exceedance fraction is $0.22$ with interval $(0.13,0.28)$, or about $160$ days. Panel (c) sets the dependence-class posterior against its prior. The data raise the atom at exact tail independence from $0.25$ to $0.40$, move the slab from $0.50$ to $0.56$ and reduce the atom at the boundary to $0.04$, so the posterior probability of asymptotic independence is $0.96$. Within the slab the posterior median of $\bar{\chi}$ is $0.28$ with interval $(0.04,0.79)$; overall, $\eta$ has posterior median $0.55$, against $0.47$ for the empirical estimate, the mean excess of $\min(E_1,E_2)$ over its $0.9$ quantile \citep{ledford1996}. Drawing the margins from their posterior, as in the cut posterior of Section~\ref{sec:inference}, moves the class probabilities by at most $0.03$; the joint posterior, through its feedback towards lighter margins, raises the atom at zero to $0.59$. Panels (d) and (e) are the posterior predictive checks of Section~\ref{sec:mcmc}. The fitted $\chi(q)$ is $0.13$, $0.07$ and $0.02$ at $q=0.9$, $0.95$ and $0.99$, with $90\%$ intervals $(0.09,0.20)$, $(0.04,0.13)$ and $(0.005,0.07)$, against empirical values $0.15$, $0.05$ and zero, and continues to $0.007$ at $q=0.999$, where no empirical estimate exists. The observed number of days on which both pollutants exceed their level-$q$ quantiles lies inside the $90\%$ posterior predictive interval at every level from $0.8$ to $0.99$, on both sides of the predictive mean: $35$ days against $39.9$ at $q=0.8$, $11$ against $9.5$ at $q=0.9$ and $2$ against $2.5$ at $q=0.95$. On a day when particulate matter exceeds its $0.9$ quantile, the chance that ozone does too is $0.13$, against $0.10$ under independence.

\begin{table}[t]
\centering
\caption{Co-pollution at Wanliu: expected number of days per $183$-day warm season on which both pollutants exceed the given marginal quantiles, under the model (posterior mean and $90\%$ credible interval), the conditional extremes model (HT) and the geometric model (WC) (estimate and $90\%$ moving-block bootstrap interval), with the observed rate over the four seasons. On the original scale the $0.90$, $0.95$ and $0.99$ quantiles are $210$, $234$ and $306\,\mu$g/m$^3$ for PM$_{2.5}$ and $248$, $273$ and $328\,\mu$g/m$^3$ for ozone.}
\label{tab:beijing}
\footnotesize
\setlength{\tabcolsep}{3.5pt}
\begin{tabular}{cc cc cc cc c}
\toprule
 & & \multicolumn{2}{c}{STAR} & \multicolumn{2}{c}{HT} & \multicolumn{2}{c}{WC} & \\
\cmidrule(lr){3-4}\cmidrule(lr){5-6}\cmidrule(lr){7-8}
PM$_{2.5}$ & O$_3$ & mean & $90\%$ interval & estimate & $90\%$ interval & estimate & $90\%$ interval & observed \\
\midrule
$0.90$ & $0.90$ & $2.4$ & $(1.6,\,3.6)$ & $2.6$ & $(1.6,\,4.0)$ & $2.7$ & $(1.2,\,4.3)$ & $2.8$ \\
$0.95$ & $0.95$ & $0.61$ & $(0.33,\,1.23)$ & $0.54$ & $(0.30,\,1.09)$ & $0.58$ & $(0.23,\,1.71)$ & $0.50$ \\
$0.99$ & $0.90$ & $0.18$ & $(0.10,\,0.41)$ & $0.11$ & $(0.06,\,0.28)$ & $0.09$ & $(0.03,\,0.25)$ & $0.00$ \\
$0.90$ & $0.99$ & $0.22$ & $(0.12,\,0.45)$ & $0.11$ & $(0.05,\,0.27)$ & $0.18$ & $(0.07,\,0.30)$ & $0.00$ \\
$0.99$ & $0.99$ & $0.04$ & $(0.01,\,0.13)$ & $0.00$ & $(0.00,\,0.07)$ & $0.02$ & $(0.00,\,0.05)$ & $0.00$ \\
\bottomrule
\end{tabular}
\end{table}

Table \ref{tab:beijing} answers the co-pollution question directly, in closed form from \eqref{eq:jointsurv} with credible intervals: both pollutants above their $0.9$ quantiles on about two and a half days per warm season, both above their $0.95$ quantiles on about three days in five seasons, one pollutant above its $0.99$ quantile while the other exceeds its $0.9$ quantile on about one day in five seasons, and both above their $0.99$ quantiles on about one day in twenty-five. Panel (f) turns the posterior into return level sets on the original scale, the curves outside which the pair falls on average once per warm season and once in five, the radius at which $P(R>r\mid W=w)=1/(183T)$ for a return period of $T$ seasons, mapped back through the marginal transformation, which is drawn from its posterior. The one-season set passes through $275\,\mu$g/m$^3$ of PM$_{2.5}$ and $300\,\mu$g/m$^3$ of ozone on the diagonal, with interval from $(253,286)$ to $(345,339)$, and the five-season set through $(324,326)$, extending along the particulate axis further than along the ozone axis, where it reaches $381$. Seven days fall outside the one-season set in the four observed seasons, against a predictive mean of $4.7$ with $90\%$ interval $(1,10)$; six lie near an axis, where the margins implied by the model have lighter tails than the standardized data, a pattern the plug-in fits show at every station (seven such days at Wanliu; Section~\ref{supp:stations}). Co-pollution at Wanliu is a moderate-level phenomenon, with an asymptotically independent joint tail and rare simultaneous extremes.

We compared these results with the conditional extremes model (HT), conditioning at the $0.9$ quantile, and the geometric model (WC), with $90\%$ intervals from a moving-block bootstrap with ten-day blocks within seasons (Table \ref{tab:beijing}). For a region lying partly below the WC threshold, the probability is the empirical frequency below the threshold plus the model probability above it. At the observed levels the methods agree, as they partly must since HT and WC draw on the data there directly: $\chi(0.9)$ is $0.13$ under the model, $0.14$ under HT and $0.15$ under WC and empirically, and $\chi(0.95)$ is $0.07$, $0.06$, $0.06$ and $0.05$. Beyond the data their point estimates part. At $q=0.99$, $\chi(q)$ is $0.020$ under the model, $0.011$ under WC, close to the independence value $0.01$, and near zero under HT, whose coefficient for ozone given extreme particulate matter is negative, $-0.35$ with interval $(-1.00,-0.09)$: on the days of heaviest particulate pollution ozone tends to fall below its median. The scale family \eqref{eq:scalefam} does not represent negative association (Section \ref{sec:parfun}). For particulate matter above its $0.99$ quantile and ozone above its $0.9$ quantile the model therefore gives the independence value, $0.001$, and HT and WC about half of it; where the association is negative, the model errs towards larger joint probabilities. No day in the four seasons falls in these regions, and the intervals of the three methods overlap at every level.

A second application, to wave height and storm surge on the North Carolina coast, a pair at the asymptotic dependence boundary, is reported in Section~\ref{supp:duck}: the classification is unambiguous and the finite-level joint probabilities are underestimated, the Student $t$ case of Section \ref{sec:s3} in real data, which the predictive check of panel (d) flags in the same way.

\subsection{Three pollutants}\label{sec:app3}

\begin{table}[t]
\centering
\caption{Trivariate fit at Wanliu. Top: residual tail dependence coefficients $\eta_C$ of the three pairs and the triple, posterior median and $90\%$ interval (identical for PM$_{2.5}$, O$_3$ and the triple, both suprema lying at the barycentre), with the empirical mean-excess estimate, and $\chi_C(q)$, posterior median with the empirical value in parentheses. Bottom: joint exceedances of marginal quantiles in days per warm season, posterior median and $90\%$ interval against the observed rate over four seasons.}
\label{tab:app3}
\footnotesize
\setlength{\tabcolsep}{3pt}
\resizebox{\textwidth}{!}{%
\begin{tabular}{l c c c ccc}
\toprule
subset $C$ & $\eta_C$ & $90\%$ interval & empirical & $\chi_C(0.9)$ (emp.) & $\chi_C(0.95)$ (emp.) & $\chi_C(0.99)$ (emp.) \\
\midrule
PM$_{2.5}$, PM$_{10}$ & 0.944 & (0.78, 1.00) & 0.728 & 0.318 (0.423) & 0.215 (0.273) & 0.102 (0.273) \\
PM$_{2.5}$, O$_3$ & 0.854 & (0.65, 1.00) & 0.473 & 0.131 (0.150) & 0.071 (0.055) & 0.021 (0.000) \\
PM$_{10}$, O$_3$ & 0.885 & (0.67, 1.00) & 0.543 & 0.154 (0.137) & 0.094 (0.082) & 0.035 (0.137) \\
all three & 0.854 & (0.65, 1.00) & 0.370 & 0.049 (0.027) & 0.026 (0.000) & 0.008 (0.000) \\
\bottomrule
\end{tabular}}

\medskip
\begin{tabular}{l ccc}
\toprule
joint exceedances, days per warm season & median & $90\%$ interval & observed \\
\midrule
all three above their $0.9$ quantiles & 0.89 & (0.63, 1.14) & 0.50 \\
PM$_{2.5}$, PM$_{10}$ above $0.95$, O$_3$ above $0.9$ & 0.38 & (0.24, 0.52) & 0.00 \\
all three above $0.95$ & 0.24 & (0.14, 0.34) & 0.00 \\
PM$_{2.5}$, PM$_{10}$ above $0.99$ & 0.19 & (0.11, 0.26) & 0.50 \\
PM$_{2.5}$, O$_3$ above $0.99$ & 0.04 & (0.02, 0.07) & 0.00 \\
all three above $0.99$ & 0.02 & (0.01, 0.03) & 0.00 \\
\bottomrule
\end{tabular}
\end{table}

\begin{figure}[t]
\centering
\includegraphics[width=\textwidth]{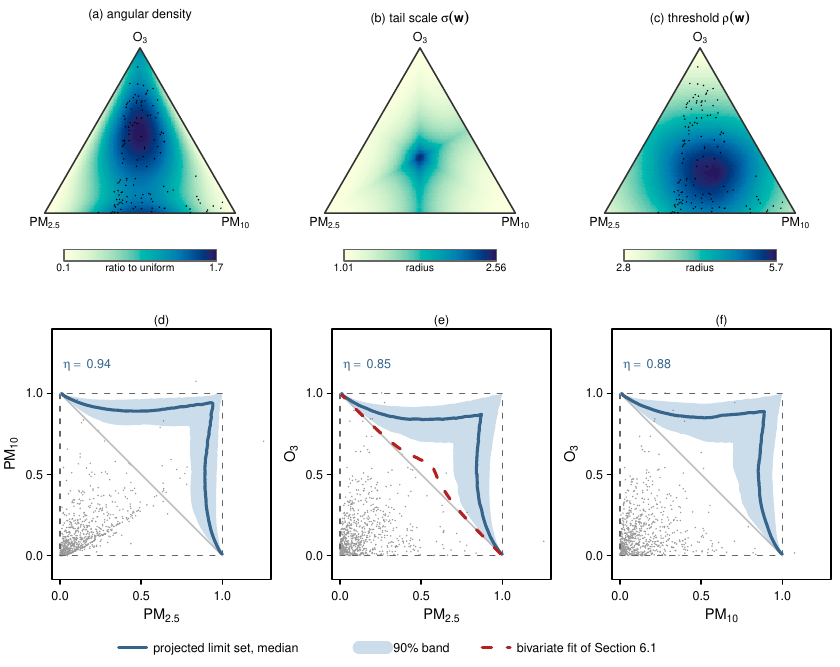}
\caption{Trivariate fit to warm-season daily maxima of PM$_{2.5}$, PM$_{10}$ and ozone at Wanliu. Top, on the simplex of angles with the three pollutants at the vertices: (a) posterior mean angular density relative to the uniform, with the angles of the $131$ days beyond the posterior median threshold (dots); (b) posterior median tail scale $\sigma(\bw)$, the radius of the limit set in each direction; (c) posterior median threshold radius $\rho(\bw)$ with the same days. Bottom, (d) to (f): pairwise scaled samples $\bE/\log n$ with the projection of the fitted limit set onto the coordinate plane, posterior median and $90\%$ band, the posterior median of the pairwise coefficient $\eta$ (Table \ref{tab:app3}), the unit box (dashed) and the independence line (grey); in (e) the posterior median limit set of the bivariate fit of Section \ref{sec:app2} (dashed red).}
\label{fig:app3}
\end{figure}

We add coarse particulate matter, PM$_{10}$, to the pair of Section \ref{sec:app2}, with the same $n=732$ days at Wanliu and the same cleaning and standardization; the generalized Pareto shapes are $0.14$, $0.27$ and $-0.20$. The two particulate series are strongly associated, with Spearman correlation $0.69$ and empirical $\chi(0.9)$ of $0.42$, and ozone is weakly associated with both, with Spearman correlations $0.23$ and $0.18$ and $\chi(0.9)$ of $0.15$ and $0.14$: a strongly dependent pair with a weakly associated third component, the configuration for which \eqref{eq:addgauge} was built.

The model has a Dirichlet angular basis of degree $4$, the Gamma body with $\log\lambda(\bw)=\sum_jb_jw_j$, the threshold of Section \ref{sec:threshold} with three anchors and $p\ge1$, and the tail scale \eqref{eq:addgauge} with uniform priors on the three $\bar{\chi}_{jk}$ and a shared relaxation rate $\nu$, twenty-six parameters in all; the class priors are not used, for the reason given in Section \ref{sec:multi}. The sampler of Section \ref{sec:mcmc} rejects a tail proposal when the envelope fails on a grid of the simplex (Section~\ref{supp:impl3}), and a chain of $40{,}000$ iterations, with $15{,}000$ of burn-in, took six minutes and gave effective sample sizes above $150$ for every parameter. In a calibration check at $n=2000$ (Section~\ref{supp:calib3}), the body, the threshold, $\chi_C(q)$ and the total deficit are recovered, while its split between pairs is weakly identified, the coefficients of weak pairs being pulled up by about $0.1$.

Table \ref{tab:app3} and Figure \ref{fig:app3} summarize the fit. Where it overlaps with Section \ref{sec:app2} it agrees with the bivariate analysis: for PM$_{2.5}$ and ozone the fitted $\chi_C(q)$ is $0.13$, $0.07$ and $0.02$ at $q=0.9$, $0.95$ and $0.99$, the bivariate values to two decimals, and the exceedance fraction is $0.20$ with interval $(0.12,0.30)$, against $0.22$. The fit locates the dependence in the particulate pair: $\bar{\chi}_{12}$ has posterior median $0.82$ against $0.11$ for PM$_{2.5}$ and ozone, and the $131$ days beyond the posterior median threshold lie away from the ozone vertex (Figure \ref{fig:app3}(a)). For the particulate pair, whose empirical $\chi(q)$ is flat at $0.42$, $0.27$ and $0.27$, the fitted values fall short, as for the asymptotically dependent coastal pair of Section~\ref{supp:duck}. The coefficients of the two ozone pairs, $0.85$ and $0.89$ with intervals reaching one, lie far above their empirical estimates and the bivariate value $0.55$, through the third property of \eqref{eq:addgauge}: the deficits of the three pairs add in the interior of the simplex, taking the tail scale at the barycentre to $2.6$ against its envelope value $3$ (Figure \ref{fig:app3}(b)) and lifting every projection of the limit set, so that its projection onto the PM$_{2.5}$-ozone plane lies well outside the bivariate limit set (Figure \ref{fig:app3}(e)). The finite-level $\chi_C(q)$ and joint exceedance counts, governed by the body and the threshold at the observed levels, are reproduced. The trivariate fit thus adds to the pairwise analyses the location of joint extremes on the simplex and the probabilities of events involving all three pollutants (Table \ref{tab:app3}, bottom), and its asymptotic coefficients carry the allocation problem of \eqref{eq:addgauge}.

\section{Discussion}\label{sec:discussion}

This paper has introduced STAR, a fully generative angular-radial model in which the boundary of the extremal region is an object of Bayesian inference, bringing together the spliced body-tail models of \citet{behrens2004}, in which the threshold is estimated, and the angular-radial and geometric representations of multivariate extremes \citep{mackay2023,nolde2022,wadsworthcampbell2024}. The splice normalizes in one dimension per angle, so the density is in closed form. The estimated threshold curve supplies the multivariate threshold selection missing from the geometric literature, and the fitted tail scale function carries a complete description of extremal dependence, with the coefficient of extremal association an explicit parameter.

The simulations indicate where STAR is most useful. The model estimates finite-level joint probabilities of compound events accurately, within and beyond the model class, and reports the dependence class as a posterior probability, to be read with the prior weights of the atoms. Its estimated threshold frees the far-tail summaries from the choice of a fixed level. At Wanliu it agrees with the conditional extremes and geometric models at observed levels and gives larger joint probabilities beyond them.

The model also has limitations. First, the dependence is estimated given the margins, the cut of Section \ref{sec:inference}, and the likelihood treats consecutive observations as independent. In the Beijing application the posterior intervals were wide enough to absorb the serial dependence; in general, a composite-likelihood adjustment \citep{ribatet2012} would account for it. Second, the margins implied by a fit decay faster than the standardized ones far in the tail (Proposition \ref{prop:slowvar}), so return level sets can be too tight near the axes. Third, outside the model class the fitted $\bar{\chi}$ reflects dependence at the observed levels, and the scale family tends to overstate it when the true scale function exceeds its cap, as for the Gaussian and inverted logistic copulas. Finally, for data with positive $\chi$, as in the coastal application of Section~\ref{supp:duck}, the model settles at its boundary, where joint probabilities far in the tail are underestimated, a situation that the diagnostics of Section~\ref{sec:mcmc} flag.

The construction holds in any dimension (Section \ref{sec:multi}), and the trivariate fit of Section \ref{sec:app3} sets out what a multivariate study has to settle before the model is used in higher dimensions: how pairwise deficits should combine in the interior of the simplex so that projections respect the bivariate fits, how their allocation between pairs is identified from a few hundred exceedances, how the class priors extend to subsets of components, and how the bases on the simplex, whose size grows combinatorially with $d$, are to be regularized. Covariate-dependent or time-varying angular densities and threshold curves, including change points in extremal dependence \citep{decarvalho2020}, and hierarchical structures for spatial replication are natural developments for environmental data. Two specifications set aside would widen the model. An unconstrained scale, with $\log\sigma(w)$ in a Bernstein basis under the envelope and endpoint conditions, would reach every scale function satisfying \eqref{eq:exactclass}, including that of the inverted logistic, at the price of constrained sampling and of $\bar{\chi}$ as an explicit parameter. An explicit asymptotically dependent component, such as an mGPD or an additive combination of gauge functions \citep{lee2025geometric}, would keep $\chi(q)$ away from zero on the scale of the data, as the Student $t$ simulations and the coastal application require. Laplace margins \citep{mackay2023} would extend the geometry to all quadrants and admit negative association, which exponential margins exclude and which the conditional extremes fit suggests at Wanliu on the days of heaviest particulate pollution.

\section*{Acknowledgments}

We thank Miguel de Carvalho for suggesting the power-mean threshold family of Section \ref{sec:threshold}.

\section*{Funding}

The first author was supported by the Spanish State Research Agency, grant PID2023-153222OB-I00.

\section*{Code availability}

The R package \texttt{star}, the scripts reproducing the simulation studies and the applications, and the trivariate implementation are available at \url{https://github.com/manueleleonelli/star}.

\bibliographystyle{plainnat}
\bibliography{bib}

\clearpage
\appendix
\renewcommand{\thesection}{S\arabic{section}}
\setcounter{table}{0}
\renewcommand{\thetable}{S\arabic{table}}
\setcounter{figure}{0}
\renewcommand{\thefigure}{S\arabic{figure}}

\begin{center}
{\LARGE Supplementary material}
\end{center}

\noindent This supplementary material states the model and its results in $d$ dimensions (Section~\ref{supp:multi}), the proofs (Section~\ref{app:proofs}), the computational details of the posterior sampler and of the trivariate implementation of Section~\ref{sec:app3} (Section~\ref{app:computation}), the simulation studies in full with a trivariate calibration check (Section~\ref{supp:sims}), and further application results (Section~\ref{supp:apps}).

\section{The model in $d$ dimensions}\label{supp:multi}

Sections~\ref{sec:model} and~\ref{sec:theory} are bivariate apart from the scale function \eqref{eq:addgauge} of Section~\ref{sec:multi}. This section records what carries over to $d$ dimensions (Section~\ref{supp:multi-comp}) and states the results of Section~\ref{sec:theory} for general $d$ (Section~\ref{supp:multi-theory}); Section~\ref{app:proofs} proves them, the bivariate results being the case $d=2$, and Sections~\ref{supp:impl3} and \ref{supp:calib3} give the trivariate implementation of Section~\ref{sec:app3} and a calibration check. Throughout, $\bm{W}=\bE/R\in\mathcal{S}_{d-1}$ with barycentric coordinates $w_1,\ldots,w_d$, $\bm{e}_j$ is the $j$th vertex of the simplex, and $g$ is the gauge of the limit set, $g(\bx)=\normone{\bx}/\sigma(\bx/\normone{\bx})$ as in Proposition~\ref{prop:gauge}.

\subsection{Components}\label{supp:multi-comp}

The decomposition \eqref{eq:factorization} and the splice \eqref{eq:splice} are already written for general $d$, with Jacobian $r^{d-1}$, and the splice still normalizes in one dimension for every angle. The joint density therefore has no $d$-dimensional normalizing constant at any $d$, the log-likelihood \eqref{eq:loglik} keeps its form, and the body or tail membership of an observation costs $O(d)$ operations. Nothing in the computational argument degrades with dimension; what grows is the number of parameters of the functions that live on the simplex.

\emph{Threshold.} The power-mean family \eqref{eq:rhop} is the $d=2$ case of
\begin{equation}\label{eq:pball}
\rho(\bw)=\Big\{\sum_{j=1}^d\Big(\frac{w_j}{u_j}\Big)^{p}\Big\}^{-1/p},
\end{equation}
the radial form of the contour $\sum_j(x_j/u_j)^p=1$, with $d$ anchor levels and one exponent, the rectangle $[0,u_1]\times\cdots\times[0,u_d]$ as $p\to\infty$ and the simplex $\sum_jx_j/u_j=1$ at $p=1$. It generalizes without change.

\emph{Body and angular density.} The Gamma body keeps $\alpha=d$ under independence (Section~\ref{sec:splice}). The Beta basis of the angular density (Section~\ref{sec:angular}), whose $K_W$ densities have degree $K_W-1$, and the Bernstein basis of \eqref{eq:parfun} become on the simplex the Dirichlet basis of degree $K$, $\{\mathrm{Dir}(\bw;\bm{k}+\bm{1}):\bm{k}\in\mathbb{N}^d,\ |\bm{k}|=K\}$, which has $\tbinom{K+d-1}{d-1}$ members, and the equal-weight mixture is the uniform density on the simplex. Both $\log\lambda(\cdot)$ and $f_{\bm{W}}$ therefore extend without change, so the angular density is the easy part of the extension in principle and the expensive one in practice: degree $9$, the degree of $K_W=10$ Beta densities, gives $55$ weights at $d=3$, $220$ at $d=4$ and $715$ at $d=5$. Low degree, as in the implementation of Section~\ref{supp:impl3}, or a sparse or low-rank prior on the weights, is needed beyond $d=3$, where the count grows fastest.

\emph{Scale function.} The two conditions of Section~\ref{sec:limitset} become $\sigma(\bm{e}_j)=1$ at the $d$ vertices and the envelope $\sigma(\bw)\le1/\max_jw_j$, the $d$-dimensional form of \eqref{eq:envelope}. Table~\ref{tab:eta3} supports Section~\ref{sec:multi}. It compares the two trivariate scale functions discussed there, the ratio $\min_jw_j/\max_jw_j$ in place of $\tau$ in \eqref{eq:scalefam} and the additive-pairwise gauge \eqref{eq:addgauge}, through the coefficients $\eta_C$ of Proposition~\ref{prop:eta3}, at the parameter values of Section~\ref{sec:multi} and Figure~\ref{fig:geometry3}. Under the ratio every pairwise coefficient equals the trivariate one, so pairwise and full-dimensional dependence are not separated; under \eqref{eq:addgauge} they differ, and the envelope holds for three equal pairs up to $\bar{\chi}_{jk}=1/2$.

Two recent alternatives could replace \eqref{eq:addgauge} in the splice, provided the vertex and envelope conditions are imposed: the piecewise-linear gauge of \citet{campbell2024}, interpolated over a triangulation of reference directions, and the neural gauges of \citet{murphybarltrop2024gauge}. The additive form keeps the $\bar{\chi}_{jk}$ as parameters, which is what the class priors need, and is what we implement.

\begin{table}[t]
\centering
\caption{Residual tail dependence coefficients $\eta_C$ implied by two trivariate scale functions with $\nu=1/2$: the ratio $\min_jw_j/\max_jw_j$ in place of $\tau$ in \eqref{eq:scalefam}, and the additive-pairwise gauge \eqref{eq:addgauge} of Section~\ref{sec:multi}. The bivariate family gives $\eta=(1+\bar{\chi})/2$: $0.75$, $0.80$, $0.95$ and $1$ at $\bar{\chi}=0.5$, $0.6$, $0.9$ and $1$. The last column reports the envelope condition, checked on a grid of the simplex.}
\label{tab:eta3}
\footnotesize
\resizebox{\textwidth}{!}{%
\begin{tabular}{ll cccc c}
\toprule
scale function & $(\bar{\chi}_{12},\bar{\chi}_{13},\bar{\chi}_{23})$ & $\eta_{12}$ & $\eta_{13}$ & $\eta_{23}$ & $\eta_{123}$ & admissible \\
\midrule
ratio $\min_jw_j/\max_jw_j$ & $\bar{\chi}=0.5$, all pairs & $0.500$ & $0.500$ & $0.500$ & $0.500$ & yes \\
ratio $\min_jw_j/\max_jw_j$ & $\bar{\chi}=0.9$, all pairs & $0.633$ & $0.633$ & $0.633$ & $0.633$ & yes \\
additive \eqref{eq:addgauge} & $(0.6,0,0)$ & $0.800$ & $0.500$ & $0.500$ & $0.444$ & yes \\
additive \eqref{eq:addgauge} & $(0.6,0.2,0.2)$ & $0.800$ & $0.632$ & $0.632$ & $0.632$ & yes \\
additive \eqref{eq:addgauge} & $(0.5,0.4,0.4)$ & $0.840$ & $0.840$ & $0.840$ & $0.840$ & yes \\
additive \eqref{eq:addgauge} & $(0.5,0.5,0.5)$ & $1$ & $1$ & $1$ & $1$ & yes, corner attained \\
additive \eqref{eq:addgauge} & $(0.55,0.55,0.55)$ & & & & & no \\
additive \eqref{eq:addgauge} & $(1,0,0)$ & $1$ & $0.500$ & $0.500$ & $0.500$ & yes \\
\bottomrule
\end{tabular}}
\end{table}

\subsection{Results}\label{supp:multi-theory}

The results of Section~\ref{sec:theory} are stated here for general $d$, with proofs in Section~\ref{app:proofs}; each reduces to its bivariate version at $d=2$, as noted after it. Throughout, $m_C(\bw)=\min_{j\in C}w_j\,\sigma(\bw)$ for nonempty $C$, with $m=m_{\{1,\ldots,d\}}$ and $m_j=m_{\{j\}}$.

\begin{propositionS}\label{prop:gauge-d}
Let $\bE$ have density \eqref{eq:factorization} with radial component \eqref{eq:splice}, $\xi\equiv0$, $f_{\bm{W}}$ positive on the interior of the simplex, $\rho(\cdot)$ bounded and $\sigma(\cdot)$ positive and continuous on the closed simplex. Then, for every $\bx$ in the interior of the orthant,
\[
-\frac{1}{t}\log f_{\bE}(t\bx)\xrightarrow[t\to\infty]{}g(\bx)=\frac{\normone{\bx}}{\sigma(\bx/\normone{\bx})},
\]
and the scaled sample cloud converges onto $G=\{\bx\ge\bm{0}:g(\bx)\le1\}$. Unit-rate margins require the envelope $\sigma(\bw)\le1/\max_jw_j$.
\end{propositionS}

Proposition~\ref{prop:gauge} is the case $d=2$.

\begin{propositionS}\label{prop:eta3}
Under the conditions of Proposition~\ref{prop:gauge-d}, for every $C$ with at least two elements, the residual tail dependence coefficient of $(E_j)_{j\in C}$ is
\[
\eta_C=\sup_{\bw\in\mathcal{S}_{d-1}}m_C(\bw),
\]
and for $\bm{\omega}$ in the interior of the simplex the angular dependence function is
\[
\Lambda(\bm{\omega})=-\lim_{t\to\infty}t^{-1}\log P(E_j>\omega_jt,\ j=1,\ldots,d)=\inf_{\bw\in\mathcal{S}_{d-1}}\frac{1}{\sigma(\bw)}\max_j\frac{\omega_j}{w_j}.
\]
Under the envelope, $\eta_C=1$ if and only if $\sigma$ attains its envelope value $1/\max_jw_j$ at a point with $\min_{j\in C}w_j=\max_jw_j$; in particular, $\eta_{1\cdots d}=1$ if and only if $\sigma(1/d,\ldots,1/d)=d$.
\end{propositionS}

The formula of Proposition~\ref{prop:eta} and the formula \eqref{eq:adf} of Proposition~\ref{prop:adf} are the case $d=2$. The statements of Propositions~\ref{prop:eta} and~\ref{prop:adf} under the family \eqref{eq:scalefam} are proved in Section~\ref{supp:proof-cor}; under \eqref{eq:addgauge} the coefficients are the optimizations of Table~\ref{tab:eta3}.

The polynomial factor of Proposition~\ref{prop:slowvar} is the one result that changes with dimension. It comes from Laplace's method at a conical minimum: a function $h$ has a conical minimum at $\bw_0$ if $h(\bw_0+\bm{v})=h(\bw_0)+\ell(\bm{v})+o(\normtwo{\bm{v}})$ as $\bm{v}\to\bm{0}$ within its domain, with $\ell$ continuous, positively homogeneous of order one and bounded below by a positive multiple of $\normtwo{\bm{v}}$ on the feasible directions.

\begin{propositionS}\label{prop:decay-d}
Assume the conditions of Proposition~\ref{prop:gauge-d} with the scale function \eqref{eq:addgauge}, and let $f_{\bm{W}}$, $\rho(\cdot)$ and $\lambda(\cdot)$ be continuous and $f_{\bm{W}}$ positive on the closed simplex. Then there are positive constants $c$ and $c_j$ such that, as $t\to\infty$,
\begin{gather*}
P(E_1>t,\ldots,E_d>t)=c\,t^{-(d-1)}e^{-t/\eta}\{1+o(1)\},\\
P(E_j>t)=c_j\,t^{-(d-1)}e^{-t}\{1+o(1)\}\quad(j=1,\ldots,d),
\end{gather*}
the first if $m$ attains its maximum $\eta=\eta_{1\cdots d}$ at a single interior point, where $1/m$ has a conical minimum, and the second if $\sigma$ attains the envelope only at the vertices of the simplex. Consequently $\chi(q)=P(E_1>t_q,\ldots,E_d>t_q)/(1-q)$ tends to zero, like $c\,\{-\log(1-q)\}^{-(d-1)}$ when $\eta=1$, and the margins of the model are lighter than exponential by the factor $t^{-(d-1)}$.
\end{propositionS}

Proposition~\ref{prop:slowvar} is the case $d=2$ under \eqref{eq:scalefam}, completed in Section~\ref{supp:proof-const}.

\FloatBarrier
\section{Proofs}\label{app:proofs}

Propositions~\ref{prop:gauge-d} to \ref{prop:decay-d} contain Propositions~\ref{prop:gauge}, \ref{prop:eta} and~\ref{prop:adf} and the expansions \eqref{eq:slowvar} of Proposition~\ref{prop:slowvar} as the case $d=2$; the statements specific to the family \eqref{eq:scalefam} are proved in Sections~\ref{supp:proof-cor} and~\ref{supp:proof-const}.

\subsection{Preliminaries}\label{supp:prelim}

Integrals over the simplex are with respect to Lebesgue measure in $(w_1,\ldots,w_{d-1})$, against which $f_{\bm{W}}$ is a density. Let $A(\bw)=1-G_{\mathrm{b}}(\rho(\bw)\mid\bw)$ be the tail mass of the splice \eqref{eq:splice}, $\bar{\rho}=\sup_{\bw}\rho(\bw)$ the largest threshold radius and $B(\bw)=f_{\bm{W}}(\bw)A(\bw)e^{\rho(\bw)/\sigma(\bw)}$. For $r>\bar{\rho}$ the tail branch of \eqref{eq:splice} gives
\begin{equation}\label{eq:survpolar}
f_{\bm{W}}(\bw)\,P(R>r\mid\bm{W}=\bw)=B(\bw)\,e^{-r/\sigma(\bw)},
\end{equation}
with $B$ positive on the interior of the simplex, integrable, since $B\le f_{\bm{W}}e^{\bar{\rho}/\inf\sigma}$, and continuous wherever $f_{\bm{W}}$, $\rho$ and $\lambda$ are. Since $E_j=RW_j$, for $\bm{a}\ge\bm{0}$ with $C=\{j:a_j>0\}$ nonempty the event $\{E_j>a_jt,\ j\in C\}$ is $\{R>t\max_{j\in C}a_j/W_j\}$, whose radial level is at least $t\max_ja_j$, so for $t\max_ja_j>\bar{\rho}$
\begin{equation}\label{eq:survivorint}
P(E_j>a_jt,\ j\in C)=\int B(\bw)\,e^{-t\Phi_{\bm{a}}(\bw)}\,d\bw,\qquad \Phi_{\bm{a}}(\bw)=\frac{1}{\sigma(\bw)}\max_{j\in C}\frac{a_j}{w_j}.
\end{equation}
The function $\Phi_{\bm{a}}$ in \eqref{eq:survivorint} is continuous where finite and infinite where $w_j=0$ for some $j\in C$. Bounding it below by its infimum everywhere, and above by the infimum plus $\varepsilon$ on the open set where this holds, which has positive measure, gives
\begin{equation}\label{eq:ldp}
-\frac{1}{t}\log P(E_j>a_jt,\ j\in C)\longrightarrow\inf_{\bw}\Phi_{\bm{a}}(\bw)=\min\{g(\bx):x_j\ge a_j,\ j\in C\},
\end{equation}
the equality because $r\bw$ satisfies the constraints exactly when $r\ge\max_{j\in C}a_j/w_j$, and $g(r\bw)=r/\sigma(\bw)$, as in \citet[Section~\ref{sec:coefficients}]{nolde2022}.

\subsection{Proof of Proposition~\ref{prop:gauge-d}}\label{supp:proof-gauge}

\emph{Gauge.} For $\bx$ in the open orthant, write $s=\normone{\bx}$ and $\bw=\bx/s$. For $ts>\bar{\rho}$, \eqref{eq:factorization} and the tail branch of \eqref{eq:splice} give
\[
f_{\bE}(t\bx)=(ts)^{-(d-1)}\,f_{\bm{W}}(\bw)\,A(\bw)\,\sigma(\bw)^{-1}\exp\big[-\{ts-\rho(\bw)\}/\sigma(\bw)\big],
\]
and since $f_{\bm{W}}(\bw)A(\bw)>0$, $-t^{-1}\log f_{\bE}(t\bx)\to s/\sigma(\bw)=g(\bx)$.

\emph{Convergence of the sample cloud.} By Proposition~\ref{prop:gauge} of \citet{nolde2022}, the cloud $\{\bE_i/\log n\}_{i\le n}$ of independent copies of $\bE$ converges onto the compact set $G=\{r\bw:0\le r\le\sigma(\bw)\}$ if, as $n\to\infty$,
\begin{equation}\label{eq:contain}
n\,P(\bE/\log n\notin U)\to0\quad\text{for every open set }U\supset G,
\end{equation}
\begin{equation}\label{eq:fill}
n\,P\{\bE/\log n\in\mathcal{B}(\bx,\varepsilon)\}\to\infty\quad\text{for every }\bx\in G\text{ and }\varepsilon>0,
\end{equation}
where $\mathcal{B}(\bx,\varepsilon)$ is the open Euclidean ball. We verify both from \eqref{eq:survpolar}, using $g(\bE/\log n)=R/\{\sigma(\bm{W})\log n\}$ and $0<\inf\sigma\le\sup\sigma<\infty$.

\emph{Proof of \eqref{eq:contain}.} As $G$ is compact, $U$ contains every point within some distance $\delta>0$ of $G$. Every point of $\{g\le1+\varepsilon'\}=(1+\varepsilon')G$ lies within $\varepsilon'\sup\sigma$ of $G$, so $U$ contains this set for $\varepsilon'=\delta/(2\sup\sigma)$, and for $n$ large enough that $(1+\varepsilon')\inf\sigma\log n>\bar{\rho}$,
\[
n\,P(\bE/\log n\notin U)\le n\,P\{R>(1+\varepsilon')\,\sigma(\bm{W})\log n\}=n^{-\varepsilon'}\int B\longrightarrow0.
\]

\emph{Proof of \eqref{eq:fill}.} For $\bx=\bm{0}$ the probability tends to one. Otherwise $\bx=s\bw_0$ with $\bw_0$ in the simplex and $0<s\le\sigma(\bw_0)$. Let $\kappa=\min(\varepsilon,s)/2$, and let $I$ be a neighbourhood of $\bw_0$ in the simplex on which $\normtwo{\bw-\bw_0}<\varepsilon/(2s)$ and $\sigma(\bw)>\sigma(\bw_0)-\kappa/2$. If $\bm{W}\in I$ and $|R/\log n-s|<\kappa$, then $\normtwo{\bE/\log n-\bx}\le|R/\log n-s|+s\,\normtwo{\bm{W}-\bw_0}<\varepsilon$, as $\normtwo{\bm{W}}\le1$. For $n$ large enough that $(s-\kappa)\log n>\bar{\rho}$ and $n^{-2\kappa/\sup\sigma}\le1/2$, \eqref{eq:survpolar} therefore gives
\[
P\{\bE/\log n\in\mathcal{B}(\bx,\varepsilon)\}\ge\int_IB(\bw)\,n^{-(s-\kappa)/\sigma(\bw)}\big\{1-n^{-2\kappa/\sigma(\bw)}\big\}\,d\bw\ge\frac{1}{2}\Big(\int_IB\Big)\,n^{-\gamma},
\]
where $\gamma=(s-\kappa)/\{\sigma(\bw_0)-\kappa/2\}<1$ because $s\le\sigma(\bw_0)$, and $\int_IB>0$ because $I$ meets the interior of the simplex in a set of positive measure, so $n$ times the bound diverges.

\emph{Envelope.} With $\bm{a}=\bm{e}_j$, $\Phi_{\bm{a}}(\bw)=1/\{w_j\sigma(\bw)\}$, and \eqref{eq:ldp} gives $-t^{-1}\log P(E_j>t)\to1/\sup_{\bw}w_j\sigma(\bw)$. Unit-rate exponential margins therefore require $w_j\sigma(\bw)\le1$ for every $\bw$ and every $j$, which is $\sigma(\bw)\le1/\max_jw_j$.

\subsection{Proof of Proposition~\ref{prop:eta3}}\label{supp:proof-eta}

\emph{Coefficients.} Since $P(E_j>t,\ j\in C)=L(e^t)\,e^{-t/\eta_C}$ with $L$ slowly varying \citep{ledford1996}, $-t^{-1}\log P(E_j>t,\ j\in C)\to1/\eta_C$. With $\bm{a}$ the indicator of $C$, $\Phi_{\bm{a}}=1/m_C$, so \eqref{eq:ldp} gives $\eta_C=\sup m_C$, and with $\bm{a}=\bm{\omega}$ it gives $\Lambda(\bm{\omega})$.

\emph{Value one.} Under the envelope, $m_C(\bw)\le\min_{j\in C}w_j/\max_jw_j\le1$, with equality if and only if $\sigma(\bw)=1/\max_jw_j$ and $\min_{j\in C}w_j=\max_jw_j$. As $m_C$ is continuous on the closed simplex, the supremum is attained, so $\eta_C=1$ if and only if both hold at some point. For $C=\{1,\ldots,d\}$ the second condition forces the barycentre, where the first reads $\sigma(1/d,\ldots,1/d)=d$.

\subsection{Proof of Proposition~\ref{prop:decay-d}}\label{supp:proof-decay}

\emph{Laplace's method.} Both survivors rest on Laplace's method at a conical minimum. Let $D\subset\mathbb{R}^m$ be bounded and polyhedral, with tangent cone $T$ at $\bw_0\in\bar{D}$, let $h$ be continuous on $D$, with a conical minimum at $\bw_0$ whose homogeneous term satisfies $\ell(\bm{v})\ge\kappa\normtwo{\bm{v}}$ on $T$ and with $\inf\{h(\bw):\normtwo{\bw-\bw_0}>\delta\}>h(\bw_0)$ for every $\delta>0$, and let $b$ be integrable on $D$ and continuous and positive at $\bw_0$. Then, as $t\to\infty$,
\begin{equation}\label{eq:laplace}
\int_Db(\bw)\,e^{-t\,h(\bw)}\,d\bw=b(\bw_0)\,t^{-m}e^{-th(\bw_0)}\int_Te^{-\ell(\bm{v})}\,d\bm{v}\,\{1+o(1)\}.
\end{equation}
For an interior point, where $T=\mathbb{R}^m$, \eqref{eq:laplace} is the leading term of the expansion of \citet{fulks1961} for a minimum at which $h$ grows linearly. Their argument localizes the integral to a ball around $\bw_0$ and rescales by $t$; since a polyhedral $D$ coincides with $\bw_0+T$ near $\bw_0$, it applies at boundary points with $\mathbb{R}^m$ replaced by $T$.

\emph{Joint survivor.} With $\bm{a}=\bm{1}$, \eqref{eq:survivorint} is $\int Be^{-t/m}$, and \eqref{eq:laplace} at the interior point $\bw^\ast$ that maximizes $m$, where $T=\mathbb{R}^{d-1}$, gives the joint survivor with $c=B(\bw^\ast)\int_{\mathbb{R}^{d-1}}e^{-\ell^\ast}$, $\ell^\ast$ being the homogeneous term of $1/m$ at $\bw^\ast$.

\emph{Marginal survivors.} With $\bm{a}=\bm{e}_j$, \eqref{eq:survivorint} is $P(E_j>t)=\int Be^{-t/m_j}$. As in Section~\ref{supp:proof-eta}, $m_j\le1$ with equality only where $\sigma$ attains the envelope and $w_j=\max_kw_k$, so if this happens only at the vertices, $1/m_j=g/w_j$ has its unique minimum $1$ at $\bm{e}_j$. Near $\bm{e}_j$, with $\bm{v}=(w_l)_{l\ne j}$ and $w_j=1-\sum_{l\ne j}v_l$, each deficit $D_{jl}$ is $O(v_l^{1+\nu_{jl}})=o(\normtwo{\bm{v}})$ because $\tau=O(v_l)$, while the deficits not involving $j$ are homogeneous of order one in $\bm{v}$. Hence $1/m_j=1+g_j(\bm{v})+o(\normtwo{\bm{v}})$, where
\[
g_j(\bm{v})=\sum_{l\ne j}v_l-\sum_{k<l,\ k,l\ne j}D_{kl}(v_k,v_l)
\]
is the gauge on the face $\{x_j=0\}$, so that $g_j(\bm{v})\ge\max_lv_l$ by the envelope. Applying \eqref{eq:laplace} at the vertex, whose tangent cone is the orthant $[0,\infty)^{d-1}$, gives the marginal survivor with $c_j=B(\bm{e}_j)\int e^{-g_j}$.

\subsection{Proof of the statements of Propositions~\ref{prop:eta} and~\ref{prop:adf} under the family \eqref{eq:scalefam}}\label{supp:proof-cor}

\emph{Proposition~\ref{prop:eta}.} By \eqref{eq:scalefam}, $\tau^{1+\nu}\le\tau$ on $(0,1)$, with equality only at $w=1/2$. With $a=\min(w,1-w)$ and $\tau=a/(1-a)$,
\[
m(w)\le a+\frac{\bar{\chi}a^2}{1-a}\le\frac{1+\bar{\chi}}{2},
\]
the second inequality being $(1-2a)\{(1-a)+\bar{\chi}(1+a)\}\ge0$, with equality only at $a=1/2$. As $m(1/2)=(1+\bar{\chi})/2$, the supremum $\eta=(1+\bar{\chi})/2$ of Proposition~\ref{prop:eta3} is attained only at the diagonal.

\emph{Proposition~\ref{prop:adf}.} By \eqref{eq:ldp}, $\Lambda(\omega)=\min\{g(\bx):\bx\ge(\omega,1-\omega)\}$. The envelope \eqref{eq:envelope} is $g(\bx)\ge\max(x_1,x_2)$, so $\Lambda(\omega)\ge\max(\omega,1-\omega)$; the point $(\omega,1-\omega)$ is feasible and $\sigma\ge1$, so $\Lambda(\omega)\le1/\sigma(\omega)\le1$; and $g$, hence $\Lambda$, is pointwise nonincreasing in $\bar{\chi}$. If $\bar{\chi}=0$ then $g(\bx)=x_1+x_2$ and $\Lambda\equiv1$; if $\bar{\chi}=1$ the feasible point $\max(\omega,1-\omega)(1,1)$ has $g=2\max(\omega,1-\omega)/\sigma(1/2)=\max(\omega,1-\omega)$. Conversely, as $\Lambda(1/2)=1/(2\eta)$ by definition and $\eta=(1+\bar{\chi})/2$, $\Lambda\equiv1$ gives $\bar{\chi}=0$ and $\Lambda(1/2)=1/2$ gives $\bar{\chi}=1$.

\subsection{Proof of Proposition~\ref{prop:slowvar}}\label{supp:proof-const}

Proposition~\ref{prop:decay-d} gives \eqref{eq:slowvar} except for the marginal survivors at $\bar{\chi}=1$.

\emph{Joint survivor.} Under \eqref{eq:scalefam} write $\eta=(1+\bar{\chi})/2$, $\bar{\nu}=\nu(1/2)$ and $s=1+3\bar{\chi}+2\bar{\chi}\bar{\nu}$. Since $\tau(1/2+u)=1-4|u|+O(u^2)$ and $\nu$ is continuous, $\sigma(1/2+u)=1+\bar{\chi}-4\bar{\chi}(1+\bar{\nu})|u|+o(u)$, so $m(1/2+u)=\eta-s|u|+o(u)$ and $1/m(1/2+u)=1/\eta+(s/\eta^2)|u|+o(u)$. By Proposition~\ref{prop:eta} this is the unique maximum of $m$, and the proof of Proposition~\ref{prop:decay-d} gives $c=B(1/2)\int e^{-(s/\eta^2)|v|}dv=2\eta^2B(1/2)/s$.

\emph{Marginal survivors.} Below the boundary the envelope is attained only at $w\in\{0,1\}$, so the marginal survivors follow, with $c_j=B(\bm{e}_j)$ because $g_j(v)=v$. At $\bar{\chi}=1$, $m_1$ also equals one at $w=1/2$, and only there besides $w=1$, with $1/m_1(1/2+u)=1+2\bar{\nu}u+o(u)$ for $u>0$ and $1+(4+2\bar{\nu})|u|+o(u)$ for $u<0$; applying \eqref{eq:laplace} on each side adds $\kappa=B(1/2)\{1/(2\bar{\nu})+1/(4+2\bar{\nu})\}$ to $c_1$ and $c_2$, giving their values in \eqref{eq:slowvar} at the boundary. The constants therefore depend on $\nu$ only through $\bar{\nu}$.

\emph{Copula.} Let $t_j(u)$ be the quantiles of the margins of the model and $\ell=-\log(1-u)$, so that $t_j(u)=\ell-\log\ell+\log c_j+o(1)$. If $\bar{\chi}<1$, the joint survivor at $\{t_1(u),t_2(u)\}$ lies between its values at the equal levels $\min_jt_j(u)$ and $\max_jt_j(u)$, both of order $\ell^{1/\eta-1}e^{-\ell/\eta}$, so the copula is asymptotically independent with coefficient $\eta$. If $\bar{\chi}=1$, then $\eta=1$ and, with $t^\ast=\max_jt_j(u)$, \eqref{eq:slowvar} gives
\[
P\{E_1>t_1(u),E_2>t_2(u)\}\ge P(E_1>t^\ast,E_2>t^\ast)=\frac{c}{\max_jc_j}\,(1-u)\{1+o(1)\},
\]
so the copula is asymptotically dependent.

\FloatBarrier
\section{Computational details}\label{app:computation}

This section documents the sampler in full: the sampling parametrization with its Jacobians, the initialization, every proposal distribution and acceptance ratio, the adaptation rule, the treatment of the three-component prior, and the computation of the derived quantities. The target is the posterior $\pi(\btheta\mid\text{data})\propto \exp\{\ell(\btheta)\}\,\pi(\btheta)$ with $\ell$ the exact log-likelihood \eqref{eq:loglik}, and every update below is a Metropolis-Hastings step whose likelihood evaluations cost $O\{n(K_W+K_\lambda+K_\nu)\}$ operations. The sampler is implemented in the R package \texttt{star}, with the likelihood, the model functions and the joint survival quadrature written in C++ through \texttt{Rcpp} and the adaptive chain written in R; on a single core, $12{,}000$ iterations at $n=1500$ and $(K_W,K_\lambda,K_\nu)=(5,2,2)$ take about $100$ seconds on a laptop.

\subsection{Sampling parametrization}\label{supp:param}

All moves take place on an unconstrained scale. Writing $\btheta^{\ast}$ for the transformed parameters, the components, the transformations and the induced log-prior contributions, Jacobians included, are as follows.

\emph{Angular weights.} The simplex weights are represented through the softmax map with the last category as reference, $\phi_k=e^{\psi_k}/(1+\sum_{j<K_W}e^{\psi_j})$ for $k<K_W$ and $\phi_{K_W}=1/(1+\sum_{j<K_W}e^{\psi_j})$, whose Jacobian is $\prod_{k=1}^{K_W}\phi_k$; under the $\mathrm{Dirichlet}(a,\ldots,a)$ prior the induced log-density is $a\sum_{k=1}^{K_W}\log\phi_k(\bm{\psi})$ up to a constant.

\emph{Body.} The shape moves on the logarithmic scale, $\alpha=e^{\tilde{\alpha}}$, with induced log-prior $\log N(\tilde{\alpha};\log d,v_\alpha^2)$, the Jacobian being absorbed by placing the Gaussian prior on $\tilde{\alpha}$ directly; the rate coefficients $\bm{b}$ are unconstrained with independent $N(0,v_\lambda^2)$ priors.

\emph{Threshold.} The anchors are $u_j=q_{\kappa_0}+o_j$ with half-normal offsets $o_j$ (Section~\ref{sec:priors}), sampled as $\tilde{o}_j=\log o_j$, with log-prior $-e^{2\tilde{o}_j}/(2s_u^2)+\tilde{o}_j$ up to a constant for each $j$; the exponent moves on the logarithmic scale, $p=e^{\tilde{p}}$, with the Gaussian prior placed on $\tilde{p}=\log p$ directly, so that the whole range $p\in(0,\infty)$, concave boundaries included, is interior to the sampling space. When the exponent is restricted to convex curves, $p\ge1$ (Section~\ref{sec:threshold}), the prior is truncated at one and proposals below it are rejected.

\emph{Tail.} The shape coefficients $\bm{z}$ are unconstrained with independent $N(0,v_\nu^2)$ priors. Within the slab, $\bar{\chi}=\mathrm{logit}^{-1}(\tilde{\chi})$, and the $\mathrm{Beta}(a_\chi,b_\chi)$ slab contributes $\log\mathrm{be}(\bar{\chi};a_\chi,b_\chi)+\log\{\bar{\chi}(1-\bar{\chi})\}$. Under the three-component prior the state carries a component label $m\in\{0,\mathrm{s},1\}$, with $\bar{\chi}$ pinned at the atoms.

\subsection{Initialization}\label{supp:init}

The chain starts at $\bm{\psi}=\bm{0}$ (uniform angular weights), $\alpha=d$, $\bm{b}=\bm{0}$ and $\bm{z}=\bm{0}$, the values implied by exact independence; $p=2$; $\bar{\chi}=0.5$ in the slab; and the anchor offsets are set so that $u_1=u_2$ matches half the empirical $0.9$ quantile of the radii, floored away from zero, which places the initial threshold curve inside the joint upper region of the sample rather than at its edge.

\subsection{Random-walk blocks}\label{supp:rw}

The chain cycles through four blocks, $\mathrm{ang}=\{\bm{\psi}\}$, $\mathrm{body}=\{\tilde{\alpha},\bm{b}\}$, $\mathrm{thr}=\{\tilde{o}_1,\tilde{o}_2,\tilde{p}\}$ and $\mathrm{tail}=\{\tilde{\chi},\bm{z}\}$. For block $B$ with current value $\btheta^{\ast}_B$ and log proposal scale $s_B$, the proposal is the spherical random walk
\[
\btheta^{\ast\prime}_B=\btheta^{\ast}_B+e^{s_B}\,\bm{\varepsilon},\qquad \bm{\varepsilon}\sim N(\bm{0},\bm{I}),
\]
all other blocks held fixed, and the acceptance probability is
\begin{equation}\label{eq:mhratio}
\alpha_B=\min\Big[1,\ \exp\big\{\ell(\btheta')-\ell(\btheta)+\log\pi^{\ast}(\btheta^{\ast\prime}_B)-\log\pi^{\ast}(\btheta^{\ast}_B)\big\}\Big],
\end{equation}
\begingroup\sloppy where $\pi^{\ast}$ collects the block's prior terms of Section \ref{supp:param}, Jacobians included; the proposal density cancels by symmetry. Explicitly, the prior increments are $a\sum_k(\log\phi_k'-\log\phi_k)$ for the angular block; $\log N(\tilde{\alpha}';\log d,v_\alpha^2)-\allowbreak\log N(\tilde{\alpha};\log d,v_\alpha^2)+\allowbreak\sum_k\{\log N(b_k';0,v_\lambda^2)-\log N(b_k;0,v_\lambda^2)\}$ for the body; $\sum_j\{-(e^{2\tilde{o}_j'}-e^{2\tilde{o}_j})/(2s_u^2)+\tilde{o}_j'-\tilde{o}_j\}$ plus $\log N(\tilde{p}';\log 2,v_p^2)-\log N(\tilde{p};\log 2,v_p^2)$ for the threshold; and $\log\mathrm{be}(\bar{\chi}';a_\chi,b_\chi)-\log\mathrm{be}(\bar{\chi};a_\chi,b_\chi)+\allowbreak\log\{\bar{\chi}'(1-\bar{\chi}')\}-\log\{\bar{\chi}(1-\bar{\chi})\}+\sum_k\{\log N(z_k';0,v_\nu^2)-\log N(z_k;0,v_\nu^2)\}$ for the tail. During burn-in each log scale is adapted by the Robbins-Monro recursion
\[
s_B\leftarrow s_B+t^{-0.6}\,(\alpha_{B,t}-0.234),
\]
with $\alpha_{B,t}$ the realized acceptance probability at iteration $t$, and is frozen thereafter \citep{roberts2009}.\par\endgroup

\subsection{Moves under the three-component prior}\label{supp:jump}

Three situations arise according to the component label. In the slab, $(\tilde{\chi},\bm{z})$ move jointly as in Section \ref{supp:rw}. At the atom $\bar{\chi}=1$, only $\bm{z}$ moves, by the same random walk. At the atom $\bar{\chi}=0$ the likelihood is free of $\bm{z}$, whose full conditional is its prior, and the update is the exact Gibbs draw $\bm{z}\sim N(\bm{0},v_\nu^2\bm{I})$.

Once per iteration a between-component move is attempted. A label $m'$ is proposed from fixed probabilities $(\beta_0,\beta_{\mathrm{s}},\beta_1)$, by default $(1/3,1/3,1/3)$; if $m'=\mathrm{s}$, a slab value is drawn from the slab prior, $\bar{\chi}'\sim\mathrm{Beta}(a_\chi,b_\chi)$, and if $m'\in\{0,1\}$ the atom value is deterministic; $\bm{z}$ is carried unchanged. Because the slab proposal density equals the slab prior density, the Beta factors cancel in the Metropolis-Hastings ratio, which reduces to
\begin{equation}\label{eq:jumpratio}
\alpha_{\mathrm{jump}}=\min\Big[1,\ \exp\{\ell(\btheta')-\ell(\btheta)\}\ \frac{\pi^{\ast}(m')}{\pi^{\ast}(m)}\ \frac{\beta_{m}}{\beta_{m'}}\Big],
\end{equation}
with $\pi^{\ast}(0)=\pi_0$, $\pi^{\ast}(1)=\pi_1$ and $\pi^{\ast}(\mathrm{s})=1-\pi_0-\pi_1$: a likelihood ratio times the prior odds of the components times the proposal correction. The move is a reversible jump with unit Jacobian, $\bar{\chi}$ entering or leaving the state through its slab prior, and $\bm{z}$ is present in every component. The scale of this move is not adapted. Posterior class probabilities are the label frequencies among retained draws.

\subsection{Derived quantities}\label{supp:derived}

For each retained draw, $\eta=(1+\bar{\chi})/2$ directly. The joint survivor \eqref{eq:jointsurv} is computed by an $M$-point quadrature on $(0,1)$, by default $M=201$, of $f_W(w)\,\bar{F}_{R\mid W}\{\max(t_1/w,\,t_2/(1-w))\mid w\}$ with the full spliced survivor, equal to $1-G_{\mathrm{b}}(\cdot\mid w)$ below the threshold and to $\{1-G_{\mathrm{b}}(\rho(w)\mid w)\}\exp\{-(\cdot-\rho(w))/\sigma(w)\}$ above it, so that $\chi(q)$ and $\bar{\chi}(q)$ follow at any $q$, coherent across the splice. The angular dependence function of Proposition~\ref{prop:adf} is evaluated as the minimum of $\max\{\omega/w,(1-\omega)/(1-w)\}/\sigma(w)$ over a grid of $999$ equispaced angles, also used for the posterior curves of $\rho(\cdot)$, $\sigma(\cdot)$, $\lambda(\cdot)$ and $f_W(\cdot)$.

\subsection{Defaults and monitoring}\label{supp:defaults}

Defaults are $30{,}000$ iterations with $10{,}000$ of burn-in and thinning by $5$, the chain length that the identification study shows the threshold block to require; basis sizes $(K_W,K_\lambda,K_\nu)=(5,2,2)$, to which the results are insensitive (Section \ref{supp:sens}); hyperparameters $a=1$, $v_\alpha=0.5$, $v_\lambda=v_\nu=1$, $v_p=0.75$, $\kappa_0=0.8$, $s_u=1$, a uniform slab and $\pi_0=\pi_1=1/4$. Convergence is monitored through rank-normalized potential scale reduction factors and effective sample sizes on the retained draws, with the log-likelihood trace as a global summary. In the verification runs of the package, realized acceptance rates settle between $0.21$ and $0.26$ for the four random-walk blocks, with between-component acceptance near $0.4$. Table \ref{tab:diag} reports four chains on each of twelve datasets per configuration at the default settings. At $n=4000$ the threshold exponent $p$ is the binding parameter, so at large sample sizes we recommend four chains of $30{,}000$ iterations with the effective sample size of $p$ monitored, lengthening the chains when it falls below a few hundred.

\begin{table}[t]
\centering
\caption{Convergence over four chains of $30{,}000$ iterations at the default settings, twelve datasets per configuration: largest potential scale reduction factor over the parameters and smallest effective sample size (mean and extreme over datasets), and mean effective sample sizes of $p$ and $\bar{\chi}$.}
\label{tab:diag}
\footnotesize
\begin{tabular}{l cc cc c c}
\toprule
 & \multicolumn{2}{c}{largest $\hat{R}$} & \multicolumn{2}{c}{smallest ESS} & ESS of $p$ & ESS of $\bar{\chi}$ \\
\cmidrule(lr){2-3}\cmidrule(lr){4-5}
Configuration & mean & max & mean & min & mean & mean \\
\midrule
STAR, $\bar{\chi}=0.6$, $n=1000$ & $1.003$ & $1.010$ & $1339$ & $522$ & $2636$ & $3445$ \\
STAR, $\bar{\chi}=0.6$, $n=4000$ & $1.008$ & $1.022$ & $492$ & $253$ & $492$ & $3492$ \\
Gaussian, $\varrho=0.9$, $n=4000$ & $1.029$ & $1.113$ & $443$ & $165$ & $443$ & $1519$ \\
\bottomrule
\end{tabular}
\end{table}

\subsection{A trivariate implementation}\label{supp:impl3}

The trivariate model of Section~\ref{sec:app3} is implemented in R, in \texttt{inst/trivariate/star3.R} of the package. Its Dirichlet angular basis has fifteen weights, fourteen free, and the priors are those of Section~\ref{sec:priors} with the defaults of Section~\ref{supp:defaults}, $\log\alpha$ centred at $\log3$, and a uniform prior on each $\bar{\chi}_{jk}$. The sampler is the adaptive Metropolis-within-Gibbs of Sections~\ref{supp:param} to \ref{supp:defaults} in five blocks, with proposal covariances estimated during burn-in, and a proposal for the tail block is rejected when the envelope fails on a grid of $3081$ points of the simplex. The coefficients $\eta_C$ of Proposition~\ref{prop:eta3} for the three pairs and the triple, $\chi_C(q)=P(E_j>t_q,\ j\in C)/(1-q)$, joint exceedance probabilities and the exceedance fraction are computed on a grid of $28{,}441$ points of the simplex.

\FloatBarrier
\section{Simulation studies in full}\label{supp:sims}

\subsection{Identification}\label{supp:ident}

Table \ref{tab:ident} reports the identification study. The Gamma shape, the anchors, the exceedance fraction and the flanks of the threshold curve contract at the parametric rate, with slopes between $-0.50$ and $-0.67$; $\rho(1/2)$, $\bar{\chi}$ and $\nu(1/2)$ contract more slowly, as Proposition~\ref{prop:slowvar} implies; and the exponent $p$ is poorly identified although the curve is. Under the prior scale $v_p=1.5$ the coverage of $\rho(1/2)$ stays near $0.8$ whatever the chain length, and $v_p=0.75$ with $30{,}000$ iterations restores $0.90$; these are the defaults. Posterior correlations across parameter blocks stay below $0.35$ under both priors, and at $n=300$ the anchoring prior mainly keeps the threshold from drifting upwards.

\begin{table}[t]
\centering
\caption{Identification study at $\bar{\chi}=0.6$ and $p=2$. (a) Slope of the logarithm of the median width of $90\%$ intervals against $\log n$, for $n\in\{500,2000,8000\}$ and $30$ replicates, with $v_p=1.5$ and $8000$ iterations. (b) Coverage of $90\%$ intervals for $\rho(1/2)$ and mean effective sample size of $\log p$ at $n=8000$, over $30$ replicates, by prior scale and chain length. (c) Largest absolute posterior correlation, always between the shape and the rate of the Gamma body, largest correlation across parameter blocks, posterior standard deviation of $\log u_1$ and exceedance fraction, under the anchoring prior ($\kappa_0=0.8$, $s_u=1$) and a weak prior ($\kappa_0=0.5$, $s_u=5$), for one dataset at $n=2000$ and five at $n=300$.}
\label{tab:ident}
\footnotesize
\begin{tabular}{cccccccccc}
\multicolumn{10}{l}{(a) contraction slopes}\\
\toprule
$\alpha$ & $u_1$ & $u_2$ & exc. & $\rho(0.15)$ & $\rho(0.85)$ & $\rho(1/2)$ & $\bar{\chi}$ & $\nu(1/2)$ & $p$ \\
\midrule
$-0.50$ & $-0.60$ & $-0.66$ & $-0.52$ & $-0.67$ & $-0.59$ & $-0.32$ & $-0.24$ & $-0.22$ & $-0.26$ \\
\bottomrule
\end{tabular}

\vspace{0.8em}
\begin{tabular}{c cc cc}
\multicolumn{5}{l}{(b) coverage of $\rho(1/2)$ at $n=8000$}\\
\toprule
 & \multicolumn{2}{c}{$8000$ iterations} & \multicolumn{2}{c}{$30{,}000$ iterations} \\
\cmidrule(lr){2-3}\cmidrule(lr){4-5}
$v_p$ & coverage & ESS & coverage & ESS \\
\midrule
$1.5$ & $0.83$ & $49$ & $0.77$ & $95$ \\
$0.75$ & $0.87$ & $59$ & $0.90$ & $149$ \\
\bottomrule
\end{tabular}\hspace{1.5em}
\begin{tabular}{l ccc ccc}
\multicolumn{7}{l}{(c) anchoring prior}\\
\toprule
 & \multicolumn{3}{c}{$n=2000$} & \multicolumn{3}{c}{$n=300$} \\
\cmidrule(lr){2-4}\cmidrule(lr){5-7}
prior & corr. & across & sd & corr. & sd & exc. \\
\midrule
anchoring & $0.50$ & $0.34$ & $0.06$ & $0.57$ & $0.19$ & $0.25$ \\
weak & $0.59$ & $0.30$ & $0.10$ & $0.53$ & $0.61$ & $0.16$ \\
\bottomrule
\end{tabular}
\end{table}

\subsection{Calibration under the model}\label{supp:calib}

Away from the boundary (Table \ref{tab:calib}), bias is small and coverage close to nominal for $\bar{\chi}$, $\chi(0.99)$ and $\Lambda(0.3)$, and a little below it, $0.82$ to $0.84$ for $\bar{\chi}$ and $\chi(0.99)$, in the asymmetric cell. The diagonal threshold $\rho(1/2)$ is biased upwards, with coverage down to $0.78$, the cost of the weakly identified exponent. At $\bar{\chi}=0.95$ the continuous prior shrinks $\bar{\chi}$ away from one and the coverage of $\chi(0.99)$ is $0.54$ at $n=1000$ and $0.78$ at $n=4000$; there the three-component prior of Section \ref{supp:class} is the tool.

\begin{table}[t]
\centering
\caption{Calibration under the model: bias, root mean squared error and $90\%$ coverage over $50$ replicates, for the coefficient of extremal association, the finite-level coefficient $\chi(0.99)$, the angular dependence function at $\omega=0.3$ and the threshold on the diagonal; asym.\ is the cell with $u_1=2.8$ and $u_2=1.7$.}
\label{tab:calib}
\footnotesize
\setlength{\tabcolsep}{3pt}
\resizebox{\textwidth}{!}{%
\begin{tabular}{cc ccc ccc ccc ccc}
\toprule
 & & \multicolumn{3}{c}{$\bar{\chi}$} & \multicolumn{3}{c}{$\chi(0.99)$} & \multicolumn{3}{c}{$\Lambda(0.3)$} & \multicolumn{3}{c}{$\rho(1/2)$} \\
\cmidrule(lr){3-5}\cmidrule(lr){6-8}\cmidrule(lr){9-11}\cmidrule(lr){12-14}
$\bar{\chi}$ & $n$ & bias & RMSE & cov. & bias & RMSE & cov. & bias & RMSE & cov. & bias & RMSE & cov. \\
\midrule
$0.15$ & $1000$ & $0.14$ & $0.20$ & $0.90$ & $0.007$ & $0.011$ & $0.88$ & $-0.010$ & $0.035$ & $1.00$ & $0.33$ & $0.46$ & $1.00$ \\
$0.45$ & $1000$ & $0.06$ & $0.19$ & $0.96$ & $0.002$ & $0.014$ & $1.00$ & $-0.002$ & $0.058$ & $0.98$ & $0.37$ & $0.49$ & $0.98$ \\
$0.75$ & $1000$ & $-0.11$ & $0.19$ & $0.98$ & $-0.018$ & $0.025$ & $0.90$ & $0.042$ & $0.071$ & $0.98$ & $0.35$ & $0.63$ & $0.94$ \\
$0.95$ & $1000$ & $-0.17$ & $0.21$ & $0.86$ & $-0.030$ & $0.036$ & $0.54$ & $0.069$ & $0.086$ & $0.86$ & $0.30$ & $0.63$ & $0.90$ \\
$0.15$ & $4000$ & $0.05$ & $0.12$ & $0.96$ & $0.002$ & $0.006$ & $0.96$ & $0.009$ & $0.024$ & $0.96$ & $0.25$ & $0.39$ & $0.90$ \\
$0.45$ & $4000$ & $0.08$ & $0.15$ & $0.90$ & $0.004$ & $0.011$ & $0.92$ & $-0.004$ & $0.041$ & $0.96$ & $0.38$ & $0.53$ & $0.82$ \\
$0.75$ & $4000$ & $0.06$ & $0.10$ & $0.98$ & $0.001$ & $0.011$ & $0.98$ & $-0.024$ & $0.042$ & $0.98$ & $0.09$ & $0.37$ & $0.92$ \\
$0.95$ & $4000$ & $-0.06$ & $0.07$ & $1.00$ & $-0.014$ & $0.017$ & $0.78$ & $0.022$ & $0.029$ & $1.00$ & $0.08$ & $0.33$ & $0.92$ \\
$0.45$, asym. & $4000$ & $0.11$ & $0.17$ & $0.84$ & $0.007$ & $0.014$ & $0.82$ & $-0.015$ & $0.050$ & $0.92$ & $0.30$ & $0.41$ & $0.78$ \\
\bottomrule
\end{tabular}}
\end{table}

\subsection{Classification of the dependence class}\label{supp:class}

Table \ref{tab:class} adds to Table~\ref{tab:rest} the sample size $n=2000$, the Gaussian copula at $\varrho=0.3$ and $0.7$, the proportion of correct decisions under the $0.5$ rule and the posterior mean of $\bar{\chi}$. The asymptotically dependent generators are classified correctly in $86$ to $100$ percent of replicates, and the model at $\bar{\chi}=0$ and $0.5$ receives probability of asymptotic dependence at most $0.23$. Across the Gaussian family the probability rises from about one half at $\varrho=0.3$ to about four fifths at $\varrho=0.9$. The posterior means in the last columns show the inflation of $\bar{\chi}$ by the scale family under the Gaussian and inverted logistic copulas (Section~\ref{sec:simulations}), which explains the probability near one half at $\varrho=0.3$; a joint tail close to the boundary over the observed range, when $\eta$ is near one, is a difficulty that classical procedures share \citep{poon2004}.

\begin{table}[t]
\centering
\caption{Classification with the three-component prior over $50$ replicates: posterior probability of asymptotic dependence, proportion of correct decisions under the $0.5$ rule, and posterior mean of $\bar{\chi}$, at $n=2000$ and $8000$. The true class is asymptotically dependent for the logistic, the Student $t$ and the model at $\bar{\chi}=1$, asymptotically independent otherwise; rows marked $^{*}$ come from a separate set of fifty replicates.}
\label{tab:class}
\footnotesize
\setlength{\tabcolsep}{5pt}
\begin{tabular}{lc cc cc cc}
\toprule
 & & \multicolumn{2}{c}{$P(\mathrm{AD}\mid\text{data})$} & \multicolumn{2}{c}{correct} & \multicolumn{2}{c}{$E(\bar{\chi}\mid\text{data})$} \\
Generator & class & $2000$ & $8000$ & $2000$ & $8000$ & $2000$ & $8000$ \\
\midrule
STAR, $\bar{\chi}=0$ & AI & $0.01$ & $0.00$ & $1.00$ & $1.00$ & $0.09$ & $0.02$ \\
STAR, $\bar{\chi}=0.5$ & AI & $0.23$ & $0.06$ & $0.88$ & $1.00$ & $0.61$ & $0.60$ \\
STAR, $\bar{\chi}=0.95$ & AI & $0.66$ & $0.74$ & $0.08$ & $0.06$ & $0.94$ & $0.97$ \\
STAR, $\bar{\chi}=1$ & AD & $0.69$ & $0.81$ & $0.86$ & $1.00$ & $0.94$ & $0.98$ \\
Gaussian, $\varrho=0.3^{*}$ & AI & $0.48$ & $0.50$ & $0.48$ & $0.50$ & $0.78$ & $0.78$ \\
Gaussian, $\varrho=0.5$ & AI & $0.68$ & $0.60$ & $0.08$ & $0.32$ & $0.92$ & $0.82$ \\
Gaussian, $\varrho=0.7^{*}$ & AI & $0.76$ & $0.68$ & $0.04$ & $0.30$ & $0.95$ & $0.85$ \\
Gaussian, $\varrho=0.9$ & AI & $0.80$ & $0.82$ & $0.02$ & $0.04$ & $0.96$ & $0.95$ \\
Inverted logistic, $\alpha_{\mathrm{IL}}=0.5$ & AI & $0.63$ & $0.72$ & $0.20$ & $0.12$ & $0.91$ & $0.94$ \\
Logistic, $\alpha_{\mathrm{L}}=0.5$ & AD & $0.87$ & $0.77$ & $0.98$ & $0.88$ & $0.98$ & $0.93$ \\
Student $t$, $(0.5,4)$ & AD & $0.88$ & $0.97$ & $1.00$ & $1.00$ & $0.99$ & $1.00$ \\
\bottomrule
\end{tabular}
\end{table}

\subsection{Misspecification}\label{supp:misspec}

Table \ref{tab:misspec} reports the continuous-prior model under four copulas outside its class, with the Hill-type estimator of $\eta$ as a benchmark. Under the Gaussian and inverted logistic copulas the errors in $\chi(q)$ are within $0.03$, while $\eta$ is overestimated by $0.16$ and $0.22$, because both true profiles exceed the cap $1+\bar{\chi}\tau(w)$ of the scale family over a range of angles; the Hill-type estimator is more accurate for $\eta$ there. Under the Frank copula the body association inflates $\eta$ and $\chi(q)$ for both estimators, and under the Student $t$ the model underestimates $\chi(0.999)$ by $0.16$, as Proposition~\ref{prop:slowvar} implies for data with positive $\chi$, while the Hill-type estimator underestimates $\eta$ by $0.19$.

\begin{table}[t]
\centering
\caption{Misspecification at $n=4000$ over $50$ replicates: mean estimation error against the true value of each quantity, with the truths in the first block; Hill is the mean excess of $\min(E_1,E_2)$ above its $0.9$ quantile. The cap of the scale family is exceeded by the Gaussian and inverted logistic profiles, the Frank copula combines body association with the joint tail of independence, and the Student $t$ carries $\chi>0$.}
\label{tab:misspec}
\footnotesize
\setlength{\tabcolsep}{3.5pt}
\begin{tabular}{l cccc c ccccc}
\toprule
 & \multicolumn{4}{c}{truth} & & \multicolumn{5}{c}{mean error} \\
Generator & $\eta$ & $\chi(0.99)$ & $\chi(0.999)$ & $\Lambda(0.3)$ & & $\eta$ & $\eta$, Hill & $\chi(0.99)$ & $\chi(0.999)$ & $\Lambda(0.3)$ \\
\midrule
Gaussian, $\varrho=0.5$ & $0.75$ & $0.13$ & $0.05$ & $0.72$ & & $0.16$ & $-0.05$ & $0.02$ & $0.03$ & $0.06$ \\
Inverted logistic, $0.5$ & $0.71$ & $0.15$ & $0.06$ & $0.76$ & & $0.22$ & $0.00$ & $0.01$ & $0.03$ & $0.00$ \\
Frank, $\theta=5$ & $0.50$ & $0.05$ & $0.005$ & $1.00$ & & $0.33$ & $0.11$ & $0.08$ & $0.06$ & $-0.17$ \\
Student $t$, $(0.5,4)$ & $1.00$ & $0.29$ & $0.26$ & $0.70$ & & $-0.01$ & $-0.19$ & $-0.14$ & $-0.16$ & $0.01$ \\
\bottomrule
\end{tabular}
\end{table}

\subsection{Comparison with competing estimators of $\chi(q)$}\label{supp:empirical}

Table~\ref{tab:empirical} gives the comparison of Section~\ref{sec:simulations} in full, on exactly the datasets of Sections \ref{supp:calib}, \ref{supp:misspec} and \ref{supp:ablation}. The empirical estimator is $N_q/\{n(1-q)\}$, with $N_q$ the number of joint exceedances of $t_q=-\log(1-q)$ and a bootstrap interval from $300$ resamples. The conditional extremes model of \citet{heffernan2004} is fitted on the Laplace scale to the exceedances of the $0.9$ quantile of each conditioning variable by the Gaussian working likelihood with empirical residuals; the joint probability is obtained by simulation, averaged over the two conditioning directions, with intervals from $100$ bootstrap resamples of the whole pipeline. Its margins are semiparametric (est.), empirical below the $0.9$ quantile and generalized Pareto above it, or the true ones (true); for the copulas, whose exact exponential margins the other estimators also use, the true margins give the like-for-like comparison of Table~\ref{tab:key}. The geometric model of \citet{wadsworthcampbell2024} is fitted with the package \texttt{geometricMVE}, with the rolling $0.95$ quantile of the radius as threshold and the gauge chosen by AIC; only its point estimates are compared, a bootstrap being too costly.

The error of the conditional model depends on its threshold (Table \ref{tab:htthresh}): under the Gaussian and inverted logistic copulas it roughly doubles at $q=0.999$ between the $0.8$ and the $0.95$ quantiles, whereas within the model at strong dependence it stays $1.2$ to $2.6$ times the model's from $\bar{\chi}=0.75$, falling as the threshold rises at $n=4000$.

\begin{table}[t]
\centering
\caption{Comparison with the conditional extremes model (HT), the geometric model (WC) and the empirical estimator (emp.) on the same datasets, over $50$ replicates. Top: within the model at $q=0.99$, true value, root mean squared error and coverage of $90\%$ intervals, posterior for the model and bootstrap otherwise, the conditional model using semiparametric (est.) or true margins. Bottom: root mean squared error under four copulas outside the model at $n=4000$, with the conditional model fitted with the true margins (HT, true), as in Table~\ref{tab:key}, and with semiparametric ones (HT, est.), conditioning on the $0.9$ quantile.}
\label{tab:empirical}
\footnotesize
\setlength{\tabcolsep}{4.5pt}
\resizebox{\textwidth}{!}{%
\begin{tabular}{cc c cccc cccc}
\toprule
 & & & \multicolumn{4}{c}{RMSE} & \multicolumn{4}{c}{coverage} \\
\cmidrule(lr){4-7}\cmidrule(lr){8-11}
$\bar{\chi}$ & $n$ & $\chi(0.99)$ & STAR & HT, est. & WC & emp. & STAR & HT, est. & HT, true & emp. \\
\midrule
$0.15$ & $1000$ & $0.010$ & $0.011$ & $0.012$ & $0.014$ & $0.024$ & $0.88$ & $0.58$ & $0.74$ & $0.06$ \\
$0.15$ & $4000$ & $0.010$ & $0.006$ & $0.005$ & $0.008$ & $0.017$ & $0.96$ & $0.74$ & $0.84$ & $0.30$ \\
$0.45$ & $1000$ & $0.029$ & $0.014$ & $0.022$ & $0.023$ & $0.055$ & $1.00$ & $0.64$ & $0.78$ & $0.28$ \\
$0.45$ & $4000$ & $0.029$ & $0.011$ & $0.013$ & $0.013$ & $0.028$ & $0.92$ & $0.70$ & $0.72$ & $0.68$ \\
$0.75$ & $1000$ & $0.062$ & $0.025$ & $0.043$ & $0.038$ & $0.080$ & $0.90$ & $0.52$ & $0.66$ & $0.24$ \\
$0.75$ & $4000$ & $0.062$ & $0.011$ & $0.022$ & $0.020$ & $0.035$ & $0.98$ & $0.74$ & $0.70$ & $0.90$ \\
$0.95$ & $1000$ & $0.092$ & $0.036$ & $0.052$ & $0.046$ & $0.083$ & $0.54$ & $0.66$ & $0.70$ & $0.54$ \\
$0.95$ & $4000$ & $0.092$ & $0.017$ & $0.033$ & $0.030$ & $0.049$ & $0.78$ & $0.66$ & $0.74$ & $0.86$ \\
\bottomrule
\end{tabular}}

\vspace{1em}
\begin{tabular}{l c c ccccc}
\toprule
Generator & $q$ & $\chi(q)$ & STAR & HT, true & HT, est. & WC & emp. \\
\midrule
Gaussian, $\varrho=0.5$ & $0.95$ & $0.244$ & $0.024$ & $0.019$ & $0.028$ & $0.025$ & $0.036$ \\
 & $0.99$ & $0.131$ & $0.025$ & $0.026$ & $0.034$ & $0.034$ & $0.057$ \\
 & $0.999$ & $0.054$ & $0.031$ & $0.024$ & $0.038$ & $0.029$ & $0.104$ \\
\addlinespace
Inverted logistic, $\alpha_{\mathrm{IL}}=0.5$ & $0.95$ & $0.289$ & $0.023$ & $0.022$ & $0.030$ & $0.032$ & $0.031$ \\
 & $0.99$ & $0.148$ & $0.021$ & $0.026$ & $0.035$ & $0.035$ & $0.058$ \\
 & $0.999$ & $0.057$ & $0.035$ & $0.024$ & $0.039$ & $0.032$ & $0.121$ \\
\addlinespace
Frank, $\theta=5$ & $0.95$ & $0.202$ & $0.076$ & $0.026$ & $0.033$ & $0.031$ & $0.034$ \\
 & $0.99$ & $0.048$ & $0.089$ & $0.022$ & $0.023$ & $0.023$ & $0.029$ \\
 & $0.999$ & $0.005$ & $0.062$ & $0.011$ & $0.011$ & $0.007$ & $0.005$ \\
\addlinespace
Student $t$, $(0.5,4)$ & $0.95$ & $0.338$ & $0.111$ & $0.024$ & $0.034$ & $0.036$ & $0.038$ \\
 & $0.99$ & $0.290$ & $0.142$ & $0.056$ & $0.074$ & $0.067$ & $0.095$ \\
 & $0.999$ & $0.262$ & $0.165$ & $0.098$ & $0.138$ & $0.109$ & $0.248$ \\
\bottomrule
\end{tabular}
\end{table}

\begin{table}[t]
\centering
\caption{Sensitivity of the conditional extremes model to its conditioning threshold: root mean squared error of $\chi(q)$ over $50$ replicates with the threshold at the $0.8$, $0.85$, $0.9$ and $0.95$ Laplace quantiles, under semiparametric (estimated) and true margins, against the model; smallest in bold. Top: within the model, at $q=0.99$, and at $q=0.999$ in the ablation setting of Section \ref{supp:ablation}; asym.\ is the asymmetric cell. Bottom: four copulas outside the model at $n=4000$.}
\label{tab:htthresh}
\footnotesize
\setlength{\tabcolsep}{4pt}
\begin{tabular}{ccc c cccc cccc}
\toprule
 & & & & \multicolumn{4}{c}{estimated margins} & \multicolumn{4}{c}{true margins} \\
\cmidrule(lr){5-8}\cmidrule(lr){9-12}
$\bar{\chi}$ & $n$ & $q$ & STAR & $0.8$ & $0.85$ & $0.9$ & $0.95$ & $0.8$ & $0.85$ & $0.9$ & $0.95$ \\
\midrule
$0.15$ & $1000$ & $0.99$ & $0.011$ & $\mathbf{0.008}$ & $0.009$ & $0.012$ & $0.012$ & $\mathbf{0.008}$ & $0.009$ & $0.010$ & $0.011$ \\
$0.45$ & $1000$ & $0.99$ & $\mathbf{0.014}$ & $0.019$ & $0.020$ & $0.022$ & $0.023$ & $0.017$ & $0.017$ & $0.020$ & $0.023$ \\
$0.75$ & $1000$ & $0.99$ & $\mathbf{0.025}$ & $0.043$ & $0.043$ & $0.043$ & $0.047$ & $0.038$ & $0.037$ & $0.038$ & $0.042$ \\
$0.95$ & $1000$ & $0.99$ & $\mathbf{0.036}$ & $0.054$ & $0.054$ & $0.052$ & $0.057$ & $0.049$ & $0.049$ & $0.043$ & $0.051$ \\
$0.15$ & $4000$ & $0.99$ & $0.006$ & $0.005$ & $0.005$ & $0.005$ & $0.008$ & $0.005$ & $\mathbf{0.004}$ & $0.005$ & $0.007$ \\
$0.45$ & $4000$ & $0.99$ & $\mathbf{0.011}$ & $0.015$ & $0.013$ & $0.013$ & $0.014$ & $0.014$ & $0.013$ & $0.012$ & $0.014$ \\
$0.75$ & $4000$ & $0.99$ & $\mathbf{0.011}$ & $0.029$ & $0.025$ & $0.022$ & $0.022$ & $0.028$ & $0.024$ & $0.021$ & $0.021$ \\
$0.95$ & $4000$ & $0.99$ & $\mathbf{0.017}$ & $0.043$ & $0.038$ & $0.033$ & $0.034$ & $0.039$ & $0.034$ & $0.027$ & $0.026$ \\
$0.45$, asym. & $4000$ & $0.99$ & $0.014$ & $0.014$ & $0.014$ & $0.016$ & $0.018$ & $\mathbf{0.013}$ & $\mathbf{0.013}$ & $\mathbf{0.013}$ & $0.016$ \\
$0.60$ & $4000$ & $0.999$ & $\mathbf{0.010}$ & $0.012$ & $0.012$ & $0.013$ & $0.018$ & $0.012$ & $0.011$ & $\mathbf{0.010}$ & $0.011$ \\
\bottomrule
\end{tabular}

\vspace{1em}
\resizebox{\textwidth}{!}{%
\begin{tabular}{lc c cccc cccc}
\toprule
 & & & \multicolumn{4}{c}{estimated margins} & \multicolumn{4}{c}{true margins} \\
\cmidrule(lr){4-7}\cmidrule(lr){8-11}
Generator & $q$ & STAR & $0.8$ & $0.85$ & $0.9$ & $0.95$ & $0.8$ & $0.85$ & $0.9$ & $0.95$ \\
\midrule
Gaussian, $\varrho=0.5$ & $0.95$ & $0.024$ & $0.028$ & $0.027$ & $0.028$ & $0.026$ & $0.020$ & $\mathbf{0.019}$ & $\mathbf{0.019}$ & $0.027$ \\
 & $0.99$ & $0.025$ & $0.031$ & $0.031$ & $0.034$ & $0.041$ & $0.023$ & $\mathbf{0.022}$ & $0.026$ & $0.033$ \\
 & $0.999$ & $0.031$ & $0.029$ & $0.031$ & $0.038$ & $0.047$ & $\mathbf{0.018}$ & $\mathbf{0.018}$ & $0.024$ & $0.033$ \\
\addlinespace
Inverted logistic, $\alpha_{\mathrm{IL}}=0.5$ & $0.95$ & $0.023$ & $0.030$ & $0.030$ & $0.030$ & $0.028$ & $\mathbf{0.022}$ & $\mathbf{0.022}$ & $\mathbf{0.022}$ & $0.023$ \\
 & $0.99$ & $\mathbf{0.021}$ & $0.032$ & $0.034$ & $0.035$ & $0.041$ & $\mathbf{0.021}$ & $0.023$ & $0.026$ & $0.034$ \\
 & $0.999$ & $0.035$ & $0.031$ & $0.035$ & $0.039$ & $0.051$ & $\mathbf{0.016}$ & $0.019$ & $0.024$ & $0.037$ \\
\addlinespace
Frank, $\theta=5$ & $0.95$ & $0.076$ & $0.029$ & $0.031$ & $0.033$ & $0.031$ & $\mathbf{0.023}$ & $0.025$ & $0.026$ & $0.028$ \\
 & $0.99$ & $0.089$ & $0.023$ & $0.022$ & $0.023$ & $0.022$ & $0.021$ & $\mathbf{0.019}$ & $0.022$ & $0.021$ \\
 & $0.999$ & $0.062$ & $0.014$ & $0.012$ & $0.011$ & $0.009$ & $0.011$ & $0.009$ & $0.011$ & $\mathbf{0.008}$ \\
\addlinespace
Student $t$, $(0.5,4)$ & $0.95$ & $0.111$ & $0.038$ & $0.035$ & $0.034$ & $0.036$ & $0.030$ & $0.027$ & $\mathbf{0.024}$ & $0.027$ \\
 & $0.99$ & $0.142$ & $0.089$ & $0.081$ & $0.074$ & $0.070$ & $0.079$ & $0.068$ & $0.056$ & $\mathbf{0.047}$ \\
 & $0.999$ & $0.165$ & $0.154$ & $0.145$ & $0.138$ & $0.140$ & $0.131$ & $0.116$ & $0.098$ & $\mathbf{0.080}$ \\
\bottomrule
\end{tabular}}
\end{table}

\subsection{The value of self-thresholding}\label{supp:ablation}

Table \ref{tab:misspecabl} extends the threshold ablation of Table~\ref{tab:rest} under the Gaussian and Frank copulas to three levels $q$, with the error and the interval width at $q=0.999$. Under the model truth the $90\%$ interval for $\eta$ widens from $0.23$, for the estimated threshold and at $\kappa=0.8$, to $0.31$ at $\kappa=0.95$. Under the Frank copula no variant covers, as the body association reaches every estimable threshold.

\begin{table}[t]
\centering
\caption{Ablation under misspecification at $n=4000$ over $50$ replicates: coverage of $90\%$ intervals for the true $\chi(q)$ at three levels, and mean error and mean interval width for $\chi(0.999)$, with the threshold estimated (full) or frozen at level $\kappa$.}
\label{tab:misspecabl}
\footnotesize
\begin{tabular}{ll ccc cc}
\toprule
Generator & threshold & cov.\ $0.95$ & cov.\ $0.99$ & cov.\ $0.999$ & error $0.999$ & width $0.999$ \\
\midrule
Gaussian, $\varrho=0.5$ & full & $0.88$ & $0.92$ & $0.78$ & $0.028$ & $0.079$ \\
 & $\kappa=0.8$ & $0.96$ & $0.70$ & $0.10$ & $0.045$ & $0.047$ \\
 & $\kappa=0.9$ & $1.00$ & $0.92$ & $0.48$ & $0.034$ & $0.061$ \\
 & $\kappa=0.95$ & $0.90$ & $0.96$ & $0.86$ & $0.020$ & $0.069$ \\
\midrule
Frank, $\theta=5$ & full & $0.28$ & $0.22$ & $0.18$ & $0.055$ & $0.101$ \\
 & $\kappa=0.8$ & $0.04$ & $0.00$ & $0.00$ & $0.107$ & $0.072$ \\
 & $\kappa=0.9$ & $0.18$ & $0.04$ & $0.00$ & $0.070$ & $0.094$ \\
 & $\kappa=0.95$ & $0.02$ & $0.22$ & $0.16$ & $0.051$ & $0.097$ \\
\bottomrule
\end{tabular}
\end{table}

\subsection{Sensitivity analyses}\label{supp:sens}

Table \ref{tab:sens} reports the sensitivity analyses, (a) and (b) refitting the same datasets under alternative settings, (c) on fifty separate replicates per generator. The anchoring prior and the basis sizes shift the exceedance fraction by at most $0.025$, $\bar{\chi}$ by at most $0.02$, against a posterior interval width near $0.61$, and $\chi(0.99)$ by under four percent on average, and an angle-varying body shape leaves the tail summaries essentially unbiased. The class probabilities follow the atom weights and depend much less on the shape of the slab; a slab concentrated towards one lowers the probability of the atom at one by $0.10$ to $0.18$, and the Gaussian copula with $\varrho=0.9$ stays above one half at $n=8000$ under both slabs, as in Table~\ref{tab:class}.

\begin{table}[t]
\centering
\caption{Sensitivity analyses. (a) Five datasets at $\bar{\chi}=0.6$ and $n=1000$: mean shift from the default fit, with the largest absolute shift over datasets in brackets, under alternative anchoring priors and basis sizes; for the body shape $\alpha(w)=1.5+1.5w$, fitted with a constant shape, mean error against the truth. (b) Mean posterior probability of asymptotic dependence over five datasets at $n=1000$, by atom weights and slab. (c) Mean posterior probability of asymptotic dependence over fifty replicates, by slab.}
\label{tab:sens}
\footnotesize
\begin{tabular}{l cccc}
\multicolumn{5}{l}{(a) shifts in the summaries}\\
\toprule
Setting & exc. & $\rho(1/2)$ & $\bar{\chi}$ & $\chi(0.99)$, \% \\
\midrule
$\kappa_0=0.7$ & $0.021$ $(0.030)$ & $-0.08$ $(0.15)$ & $0.004$ $(0.069)$ & $1.6$ $(8.2)$ \\
$\kappa_0=0.9$ & $-0.025$ $(0.044)$ & $0.07$ $(0.55)$ & $0.005$ $(0.028)$ & $-0.2$ $(15.7)$ \\
$s_u=0.5$ & $0.023$ $(0.045)$ & $-0.15$ $(0.34)$ & $0.018$ $(0.054)$ & $3.4$ $(13.4)$ \\
$s_u=2$ & $-0.024$ $(0.054)$ & $0.13$ $(0.74)$ & $0.001$ $(0.036)$ & $-1.2$ $(8.2)$ \\
$(K_W,K_\lambda,K_\nu)=(3,1,1)$ & $-0.005$ $(0.019)$ & $0.06$ $(0.23)$ & $0.006$ $(0.017)$ & $-2.1$ $(10.2)$ \\
$(K_W,K_\lambda,K_\nu)=(8,4,4)$ & $0.008$ $(0.024)$ & $-0.01$ $(0.27)$ & $-0.001$ $(0.006)$ & $2.4$ $(14.3)$ \\
body shape $\alpha(w)$, error & $0.010$ & $-0.03$ & $0.016$ & $2.8$ \\
\bottomrule
\end{tabular}

\vspace{0.8em}
\begin{tabular}{l cc}
\multicolumn{3}{l}{(b) atom weights and slab}\\
\toprule
Prior & $\bar{\chi}=0.5$ & $\bar{\chi}=0.95$ \\
\midrule
$\pi_0=\pi_1=0.1$ & $0.09$ & $0.24$ \\
$\pi_0=\pi_1=0.25$ & $0.24$ & $0.54$ \\
$\pi_0=\pi_1=0.4$ & $0.41$ & $0.82$ \\
$\mathrm{Beta}(2,2)$ slab & $0.24$ & $0.57$ \\
\bottomrule
\end{tabular}\hspace{1.5em}
\begin{tabular}{l cc cc}
\multicolumn{5}{l}{(c) concentrated slab}\\
\toprule
 & \multicolumn{2}{c}{$n=2000$} & \multicolumn{2}{c}{$n=8000$} \\
\cmidrule(lr){2-3}\cmidrule(lr){4-5}
Generator & uniform & $\mathrm{Beta}(2,0.8)$ & uniform & $\mathrm{Beta}(2,0.8)$ \\
\midrule
Gaussian, $\varrho=0.9$ & $0.82$ & $0.64$ & $0.76$ & $0.66$ \\
STAR, $\bar{\chi}=0.95$ & $0.66$ & $0.51$ & $0.74$ & $0.61$ \\
STAR, $\bar{\chi}=1$ & $0.70$ & $0.53$ & $0.79$ & $0.62$ \\
\bottomrule
\end{tabular}
\end{table}

\subsection{Plug-in cut, cut posterior and joint posterior}\label{supp:cutjoint}

The main paper uses the plug-in cut for the dependence summaries and the cut posterior for summaries on the original scale (Section~\ref{sec:inference}). This experiment compares them with the joint posterior, which updates the tail parameters $(\sigma_j,\xi_j)$ of the semiparametric margins given $\btheta$, on the same data with the same dependence sampler. Data are generated from the model at $\bar{\chi}=0.6$ and pushed through generalized Pareto margins of scale $1$ and shape $0.3$, so that the marginal model of Section~\ref{sec:app2} is exact and the standardized data carry the copula of the model on exact exponential margins. The targets are the values $0.137$ and $0.084$ of $\chi(0.95)$ and $\chi(0.99)$ for that copula, computed at the marginal quantiles of the model, the shape $\xi_1=0.3$, and two probabilities on the original scale, with $y^\ast$ the true $0.995$ marginal quantile: the joint exceedance probability $p^\ast=P(Y_1>y^\ast,Y_2>y^\ast)=3.5\times10^{-4}$ and the marginal probability $P(Y_1>y^\ast)=0.005$, both computed from \eqref{eq:jointsurv}. Twenty replicates are drawn at $n=2000$ and at $n=400$, with $200$ and $40$ exceedances per margin. The cut posterior uses ten marginal draws and the joint posterior adds a random-walk block per margin.

Table~\ref{tab:cutjoint} gives the result. The three treatments agree on the copula, with coverage of $\chi(0.95)$ and $\chi(0.99)$ from $0.90$ to $1.00$ in every arm at both sample sizes. On the original scale the plug-in cut covers $p^\ast$ in $0.70$ of the replicates at $n=2000$ and $0.40$ at $n=400$, the cut posterior in $0.85$ and $0.75$, and the joint posterior in $0.90$ at both. The margins implied by a fit are lighter than exponential (Proposition~\ref{prop:decay-d}), so under both cuts the fitted $P(Y_1>y^\ast)$ falls short of $0.005$, by about a quarter at $n=2000$ and almost a half at $n=400$; the joint posterior narrows the gap by raising $\xi_j$ by about $0.06$, and overshoots at $n=2000$, where $p^\ast$ is overstated by $17\%$. The same mismatch shows in the class probabilities: the copula of the model on exact exponential margins receives $P(\mathrm{AD}\mid\text{data})=0.64$ under the plug-in cut at $n=2000$, against $0.36$ for the same samples on the margins of the model, and $0.47$ against $0.29$ at $n=400$, since without the marginal factor of Proposition~\ref{prop:decay-d} the extremes sit further out.

\begin{table}[t]
\centering
\caption{Plug-in cut, cut posterior and joint posterior over margins and dependence, with generalized Pareto margins: bias of the posterior median, coverage and mean width of $90\%$ intervals, over twenty replicates at each sample size. Targets are $\chi(0.95)$ and $\chi(0.99)$ of the copula of the model, the joint exceedance probability $p^\ast$ of the true $0.995$ marginal quantiles, the marginal probability $P(Y_1>y^\ast)$ implied by the fit, and the shape $\xi_1$ of the first margin.}
\label{tab:cutjoint}
\footnotesize
\setlength{\tabcolsep}{3.5pt}
\resizebox{\textwidth}{!}{%
\begin{tabular}{l c ccc ccc ccc}
\toprule
 & & \multicolumn{3}{c}{plug-in cut} & \multicolumn{3}{c}{cut posterior} & \multicolumn{3}{c}{joint posterior} \\
\cmidrule(lr){3-5}\cmidrule(lr){6-8}\cmidrule(lr){9-11}
quantity & truth & bias & coverage & width & bias & coverage & width & bias & coverage & width \\
\midrule
\multicolumn{11}{l}{$n=2000$, 200 exceedances per margin; $P(\mathrm{AD}\mid\text{data})$ 0.64, 0.62 and 0.52} \\
$\chi(0.95)$ & 0.137 & $-$0.007 & 0.90 & 0.059 & $-$0.008 & 1.00 & 0.061 & $-$0.012 & 0.95 & 0.067 \\
$\chi(0.99)$ & 0.084 & $-$0.005 & 0.90 & 0.055 & $-$0.006 & 1.00 & 0.057 & $-$0.010 & 0.95 & 0.065 \\
$10^4p^\ast$ & 3.49 & $-$0.48 & 0.70 & 2.31 & $-$0.50 & 0.85 & 3.36 & 0.58 & 0.90 & 4.83 \\
$10^3P(Y_1>y^\ast)$ & 5.00 & $-$1.41 & 0.25 & 1.58 & $-$1.29 & 0.65 & 3.12 & $-$0.33 & 0.95 & 3.86 \\
$\xi_1$ & 0.300 & $-$0.033 & 0.90 & 0.304 & $-$0.026 & 0.95 & 0.300 & 0.035 & 0.95 & 0.345 \\
\midrule
\multicolumn{11}{l}{$n=400$, 40 exceedances per margin; $P(\mathrm{AD}\mid\text{data})$ 0.47, 0.45 and 0.40} \\
$\chi(0.95)$ & 0.137 & $-$0.024 & 0.95 & 0.104 & $-$0.027 & 0.90 & 0.103 & $-$0.030 & 0.95 & 0.111 \\
$\chi(0.99)$ & 0.084 & $-$0.020 & 0.95 & 0.087 & $-$0.023 & 0.90 & 0.086 & $-$0.027 & 0.95 & 0.091 \\
$10^4p^\ast$ & 3.49 & $-$1.52 & 0.40 & 2.90 & $-$1.72 & 0.75 & 4.57 & $-$0.79 & 0.90 & 7.14 \\
$10^3P(Y_1>y^\ast)$ & 5.00 & $-$2.31 & 0.25 & 2.15 & $-$2.14 & 0.75 & 5.69 & $-$1.01 & 0.80 & 7.40 \\
$\xi_1$ & 0.300 & $-$0.042 & 0.75 & 0.741 & $-$0.028 & 0.80 & 0.657 & 0.020 & 0.80 & 0.723 \\
\bottomrule
\end{tabular}}
\end{table}

\subsection{Calibration of the trivariate implementation}\label{supp:calib3}

Table~\ref{tab:calib3} reports twenty replicates at $n=2000$ from the model with $\bar{\chi}=(0.6,0.2,0.2)$, $\nu=0.5$, $u_j=2.2$, $p=2$, $\alpha=3$, $\bm{b}=(0.1,-0.1,0)$ and an angular density tilted towards the $(1,2)$ edge. Body and threshold are recovered as in the bivariate cells, with biases below a quarter of the interval width and coverages between $0.80$ and $0.95$; the exceedance fraction is unbiased, and $\chi_C(q)$ close to unbiased with coverages of $0.75$ to $0.90$. Of the dependence parameters, $\eta_{12}$ of the strong pair is unbiased, while the coefficients of the two weak pairs are pulled up by about $0.1$, $\bar{\chi}_{12}$ is pulled down as $\bar{\chi}_{13}$ is pulled up, and the intervals for the $\bar{\chi}_{jk}$ are about $0.6$ wide, so that the split of the total deficit between pairs is weakly identified at this sample size (Section~\ref{sec:multi}).

\begin{table}[t]
\centering
\caption{Calibration check of the trivariate implementation: twenty replicates at $n=2000$ from the model with $\bar{\chi}=(0.6,0.2,0.2)$, $\nu=0.5$, $u_j=2.2$, $p=2$ and $\alpha=3$; truth, mean of the posterior medians, bias, root mean squared error, coverage and mean width of $90\%$ intervals. Chains of $20{,}000$ iterations with $8{,}000$ of burn-in.}
\label{tab:calib3}
\footnotesize
\begin{tabular}{l cccccc}
\toprule
quantity & truth & mean of medians & bias & RMSE & coverage & width \\
\midrule
$\eta_{12}$ & 0.800 & 0.787 & $-$0.013 & 0.067 & 1.00 & 0.359 \\
$\eta_{13}$ & 0.632 & 0.742 & 0.110 & 0.140 & 0.80 & 0.368 \\
$\eta_{23}$ & 0.632 & 0.731 & 0.100 & 0.124 & 0.80 & 0.371 \\
$\eta_{123}$ & 0.632 & 0.695 & 0.064 & 0.103 & 0.95 & 0.401 \\
$\chi_{12}(0.99)$ & 0.0190 & 0.0144 & $-$0.0046 & 0.0078 & 0.75 & 0.0213 \\
$\chi_{13}(0.99)$ & 0.0055 & 0.0093 & 0.0038 & 0.0068 & 0.80 & 0.0155 \\
$\chi_{123}(0.9)$ & 0.0126 & 0.0133 & 0.0008 & 0.0034 & 0.90 & 0.0106 \\
$\chi_{123}(0.99)$ & 0.0009 & 0.0013 & 0.0005 & 0.0009 & 0.85 & 0.0031 \\
exceedance fraction & 0.386 & 0.381 & $-$0.005 & 0.025 & 0.80 & 0.069 \\
$\rho$ at the barycentre & 3.811 & 3.854 & 0.043 & 0.391 & 0.95 & 0.694 \\
$u_1$ & 2.200 & 2.254 & 0.054 & 0.129 & 0.80 & 0.470 \\
$p$ & 2.000 & 2.090 & 0.090 & 0.775 & 0.90 & 1.292 \\
$\bar{\chi}_{12}$ & 0.600 & 0.455 & $-$0.145 & 0.231 & 0.90 & 0.633 \\
$\bar{\chi}_{13}$ & 0.200 & 0.360 & 0.160 & 0.244 & 0.85 & 0.631 \\
\bottomrule
\end{tabular}
\end{table}

\FloatBarrier
\section{Further application results}\label{supp:apps}

\subsection{The twelve Beijing stations}\label{supp:stations}

Table \ref{tab:stations} repeats the analysis of Section~\ref{sec:app2} at all twelve stations of \citet{zhang2017}, with the same cleaning rule, standardization, priors and chain length. The cleaning rule removes $19$ ozone and $3$ particulate hours at Dongsi, among them ozone readings above $1000\,\mu$g/m$^3$ in August 2013 while every other station read below $180$, one particulate hour at Wanliu, and at most ten hours per pollutant elsewhere. The results are consistent across the city. Every station favours asymptotic independence, with posterior probability between $0.74$ and $0.96$, and the finite-level summaries vary little: the exceedance fraction lies between $0.14$ and $0.24$, $\chi(0.9)$ between $0.13$ and $0.21$, against empirical values of $0.10$ to $0.20$, and $\chi(0.99)$ between $0.02$ and $0.07$, the smallest values belonging to the four stations with the largest weight on exact independence. The posterior predictive checks hold at every station: the observed number of days on which both pollutants exceed their $0.9$ or $0.95$ quantiles lies inside the $90\%$ predictive interval in all $24$ comparisons.

Three tendencies are common to the stations. The observed counts lie mostly below the predictive mean, $117$ days against $151$ at the $0.9$ level and $22$ against $44$ at the $0.95$ level, though the stations share their weather; the posterior median of $\eta$, between $0.53$ and $0.81$, exceeds the empirical estimate, between $0.41$ and $0.53$, as under the Gaussian and inverted logistic copulas of Section~\ref{supp:misspec}, where the finite-level summaries remained accurate; and the days outside the one-season return level set of the plug-in fit lie mostly near the axes, $64$ against $20$ expected, where the margins implied by the model are lighter than the standardized data, as Proposition~\ref{prop:slowvar} predicts.

\begin{table}[t]
\centering
\caption{The twelve Beijing stations, warm-season daily maxima of PM$_{2.5}$ and ozone: sample size; posterior probabilities of the three components; posterior median of $\eta$ with the empirical estimate of Section~\ref{sec:app2}; posterior median exceedance fraction; posterior mean of $\chi(0.9)$ with its empirical value, and of $\chi(0.99)$; days on which both pollutants exceed their $0.95$ quantiles, observed and predictive mean with $90\%$ interval; and days outside the plug-in fit's one-season return level set with $w<0.25$ or $w>0.75$, observed and expected.}
\label{tab:stations}
\footnotesize
\setlength{\tabcolsep}{3pt}
\resizebox{\textwidth}{!}{%
\begin{tabular}{l c ccc cc c cc c cc cc}
\toprule
 & & \multicolumn{3}{c}{class} & \multicolumn{2}{c}{$\eta$} & & \multicolumn{2}{c}{$\chi(0.9)$} & & \multicolumn{2}{c}{joint exceedances} & \multicolumn{2}{c}{near the axes} \\
\cmidrule(lr){3-5}\cmidrule(lr){6-7}\cmidrule(lr){9-10}\cmidrule(lr){12-13}\cmidrule(lr){14-15}
Station & $n$ & $\bar{\chi}=0$ & slab & $\bar{\chi}=1$ & STAR & emp. & exc. & STAR & emp. & $\chi(0.99)$ & obs. & predicted & obs. & exp. \\
\midrule
Aotizhongxin & $730$ & $0.38$ & $0.54$ & $0.08$ & $0.58$ & $0.44$ & $0.17$ & $0.154$ & $0.096$ & $0.026$ & $2$ & $3.1$ $(0,\,7)$ & $4$ & $1.8$ \\
Changping & $727$ & $0.18$ & $0.63$ & $0.20$ & $0.77$ & $0.48$ & $0.14$ & $0.199$ & $0.151$ & $0.056$ & $2$ & $4.4$ $(1,\,9)$ & $6$ & $1.6$ \\
Dingling & $719$ & $0.23$ & $0.60$ & $0.17$ & $0.72$ & $0.50$ & $0.14$ & $0.204$ & $0.195$ & $0.051$ & $2$ & $4.4$ $(1,\,9)$ & $5$ & $1.5$ \\
Dongsi & $727$ & $0.11$ & $0.65$ & $0.23$ & $0.81$ & $0.46$ & $0.17$ & $0.191$ & $0.138$ & $0.055$ & $2$ & $4.4$ $(1,\,9)$ & $4$ & $1.7$ \\
Guanyuan & $727$ & $0.23$ & $0.67$ & $0.10$ & $0.67$ & $0.50$ & $0.22$ & $0.166$ & $0.138$ & $0.037$ & $2$ & $3.6$ $(0,\,8)$ & $6$ & $1.7$ \\
Gucheng & $721$ & $0.18$ & $0.68$ & $0.13$ & $0.71$ & $0.50$ & $0.21$ & $0.150$ & $0.125$ & $0.038$ & $3$ & $3.3$ $(0,\,8)$ & $5$ & $1.8$ \\
Huairou & $718$ & $0.21$ & $0.69$ & $0.10$ & $0.68$ & $0.49$ & $0.24$ & $0.187$ & $0.097$ & $0.048$ & $2$ & $3.9$ $(1,\,8)$ & $6$ & $1.6$ \\
Nongzhanguan & $732$ & $0.40$ & $0.52$ & $0.07$ & $0.56$ & $0.46$ & $0.16$ & $0.162$ & $0.123$ & $0.027$ & $0$ & $3.1$ $(0,\,8)$ & $5$ & $1.7$ \\
Shunyi & $701$ & $0.17$ & $0.57$ & $0.26$ & $0.81$ & $0.50$ & $0.14$ & $0.208$ & $0.157$ & $0.065$ & $2$ & $4.6$ $(1,\,9)$ & $7$ & $1.6$ \\
Tiantan & $732$ & $0.43$ & $0.50$ & $0.07$ & $0.53$ & $0.41$ & $0.15$ & $0.162$ & $0.096$ & $0.026$ & $0$ & $3.1$ $(0,\,7)$ & $4$ & $1.7$ \\
Wanliu & $732$ & $0.40$ & $0.56$ & $0.04$ & $0.55$ & $0.47$ & $0.22$ & $0.130$ & $0.150$ & $0.022$ & $2$ & $2.5$ $(0,\,6)$ & $7$ & $1.8$ \\
Wanshouxigong & $731$ & $0.15$ & $0.65$ & $0.20$ & $0.76$ & $0.53$ & $0.17$ & $0.176$ & $0.150$ & $0.051$ & $3$ & $3.9$ $(1,\,8)$ & $5$ & $1.7$ \\
\bottomrule
\end{tabular}}
\end{table}

\subsection{Wave height and storm surge on the North Carolina coast}\label{supp:duck}

The Beijing pair is asymptotically independent. This second application shows how the model behaves at the other end, for a pair with positive $\chi$, where Proposition~\ref{prop:slowvar} limits what the $\xi=0$ splice can represent. The data are daily maxima of significant wave height at National Data Buoy Center buoy 44100 and of storm surge, the observed water level less its astronomical prediction, at the co-located National Ocean Service tide gauge 8651370, at Duck, North Carolina, from 2010 to 2023, with $n=4485$. Waves and surge on this coast are driven by the same storms, and the empirical $\chi(q)$ is flat, $0.56$, $0.52$ and $0.58$ at $q=0.9$, $0.95$ and $0.99$, the signature of asymptotic dependence. The margins are standardized as in Section~\ref{sec:app2}, with generalized Pareto shapes $-0.01$ for wave height and $-0.11$ for surge.

The model identifies the class correctly: the posterior probability of asymptotic dependence is $0.94$, $\bar{\chi}$ sits at the atom at one, and the angular dependence function lies on its boundary curve $\max(\omega,1-\omega)$ of Proposition~\ref{prop:adf}. The threshold is well behaved, with exceedance fraction $0.13$ and interval $(0.09,0.20)$, so the threshold diagnostic of Section~\ref{sec:mcmc} does not fire. At the atom the $\xi=0$ splice makes $\chi(q)$ decay on the scale of the data, as Proposition~\ref{prop:slowvar} requires: the fitted $\chi(q)$ is $0.36$, $0.29$ and $0.19$ at the three levels, below the flat empirical values, and the joint exceedance probabilities of Table~\ref{tab:duck} are underestimated by a factor of about one and a half at the $0.9$ level and three at the $0.99$ level, $13$ against $21$ days per year and $0.7$ against $2.1$. This is the Student $t$ case of Tables~\ref{tab:misspec} and \ref{tab:empirical} in real data. The comparison of fitted and empirical $\chi(q)$, available from any fit, flags it; for such a pair the model supplies the class probability, and the joint probabilities call for an asymptotically dependent construction of the kind discussed in Section~\ref{sec:discussion}.

\begin{table}[t]
\centering
\caption{Joint exceedance probabilities for daily maximum wave height and storm surge at Duck: posterior mean and $90\%$ interval under the fitted model, the empirical frequency, and the implied number of days per year under each. Levels are marginal quantiles. The shortfall is the cost of the exponential-tail splice on an asymptotically dependent pair of variables.}
\label{tab:duck}
\footnotesize
\setlength{\tabcolsep}{4pt}
\begin{tabular}{cc cc c cc}
\toprule
wave & surge & STAR & $90\%$ interval & empirical & days/yr, STAR & days/yr, empirical \\
\midrule
$0.90$ & $0.90$ & $0.036$ & $(0.033,\,0.039)$ & $0.056$ & $13.2$ & $20.5$ \\
$0.95$ & $0.95$ & $0.014$ & $(0.013,\,0.016)$ & $0.026$ & $5.3$ & $9.5$ \\
$0.99$ & $0.90$ & $0.005$ & $(0.005,\,0.006)$ & $0.011$ & $1.9$ & $4.1$ \\
$0.90$ & $0.99$ & $0.005$ & $(0.005,\,0.006)$ & $0.010$ & $1.9$ & $3.6$ \\
$0.99$ & $0.99$ & $0.002$ & $(0.002,\,0.002)$ & $0.006$ & $0.7$ & $2.1$ \\
\bottomrule
\end{tabular}
\end{table}

\end{document}